\documentclass[a4paper,11pt]{article} 
  
\usepackage[english]{babel}
\usepackage[a4paper]{geometry}
\usepackage{amsmath}
\usepackage{amsthm}
\usepackage{amssymb}
\usepackage{bbm}
\usepackage{xcolor}
\usepackage{enumerate}
\usepackage{dsfont}
\usepackage{cancel}
\usepackage{ulem} 

\usepackage{hyperref}         

\usepackage{tikz}
\usetikzlibrary{fit}
\usetikzlibrary{shapes.geometric}
\usetikzlibrary{decorations.pathmorphing}

\tikzset{
point/.style={circle,fill=black,inner sep=1pt},
vertex/.style={circle,fill=black,inner sep=1.5pt},   
bvertex/.style={circle,fill=black,inner sep=2.8pt},
Bvertex/.style={circle,fill=black,inner sep=4pt}, 
specialEP/.style={rectangle,fill=white,draw,inner sep=3pt},  
whitevex/.style={circle,fill=white,draw, inner sep=2pt},
linelabel/.style={sloped,above,very near start, inner sep=1pt,execute at begin node=$\scriptstyle,execute at end node=$},
baseline=(current  bounding  box.center),doubled/.style={double distance= 1pt,line width=1.5pt},
th/.style={line width=0.5 pt, gray},  
med/.style={line width=1 pt}  
}

\xdefinecolor{myred}{rgb}{0.7,0,0.2} 
\xdefinecolor{mypurple}{rgb}{0.6,0.2,0.4} 
\xdefinecolor{myblue}{rgb}{0.259,0.2,0.6}   
\xdefinecolor{myorange}{rgb}{1,0.6,0.2}   
\definecolor{orange}{rgb}{1,0.5,0}
\xdefinecolor{purpleish}{cmyk}{0.75,0.75,0,0}

\def\bR{\mathbb{R}}

\def\bN{\mathbb{N}}

\def\bZ{\mathbb{Z}}

\def\cV{\mathcal{V}}

\def\cF{\mathcal{F}}

\def\cE{\mathcal{E}}
\def\cK{\mathcal{K}}
\def\cH{\mathcal{H}}
\def\cR{\mathcal{R}}

\def\eps{\varepsilon}
\def\ph{\varphi}

\def\indic{\hbox{\raise-2pt \hbox{\indbf 1}}}

\let\==\equiv

\let\io=\infty
\let\0=\noindent

\def\*{{\hfill\break\null\hfill\break}}

\def\tende#1{\,\vtop{\ialign{##\crcr\rightarrowfill\crcr
             \noalign{\kern-1pt\nointerlineskip}
             \hskip3.pt${\scriptstyle #1}$\hskip3.pt\crcr}}\,}
\def\otto{\,{\kern-1.truept\leftarrow\kern-5.truept\to\kern-1.truept}\,}

\newtheorem{theorem}{Theorem}[section]  

\newtheorem{proposition}[theorem]{Proposition}
\newtheorem{lemma}[theorem]{Lemma}

\numberwithin{equation}{section}

\def\be{\begin{equation}}
\def\ee{\end{equation}}

     \let\g=\gamma     \let\d=\delta     
       \let\th=\vartheta      
                          
\let\s=\sigma          \let\ph=\varphi   
   \let\o=\omega     
        \let\L=\Lambda

\definecolor{lightblue}{rgb}{0, 0.33, 0.71}

\def\aa{\mathfrak{a}}

\def \blue#1 {\textcolor{blue}{#1}}
\def \red#1 {\textcolor{red}{#1}}
\def \blue#1{\textcolor{blue}{#1}}

\DeclareFontFamily{U}{mathx}{\hyphenchar\font45}
\DeclareFontShape{U}{mathx}{m}{n}{
      <5> <6> <7> <8> <9> <10>
      <10.95> <12> <14.4> <17.28> <20.74> <24.88>
      mathx10
      }{}
\DeclareSymbolFont{mathx}{U}{mathx}{m}{n}
\DeclareFontSubstitution{U}{mathx}{m}{n}
\DeclareMathAccent{\widecheck}{0}{mathx}{"71}

\def\bskip{\\[-0.6cm]}

\title{Lee-Huang-Yang Energy for Dilute Bose Gases: \\ an Upper Bound for General Potentials} 

\date{\today}

\author{Giulia Basti\footnote{Department of Mathematics ``Guido Castelnuovo'', La Sapienza, Piazzale Aldo Moro, 5, 00185 Roma}, \, Morris Brooks\footnote{Institute of Mathematics, University of Zurich, Winterthurerstrasse 190, 8057 Zurich}, \, Serena Cenatiempo\footnote{Gran Sasso Science Institute, Viale Francesco Crispi 7, 67100 L'Aquila}, \, Alessandro Olgiati\footnote{Dipartimento di Matematica, Politecnico di Milano, Piazza Leonardo da Vinci 32, 20133 Milano}, \\ Benjamin Schlein$^\dagger$}

\begin{document}

\maketitle 

\begin{abstract}
We establish an upper bound for the ground state energy density of a dilute Bose gas, interacting through a general short-range, repulsive potential. To reach this goal, we extend the approach recently introduced in \cite{BBCOS} for hard spheres. 
\end{abstract}

\section{Introduction}

Based on the fundamental work of Bogoliubov \cite{B}, Lee-Huang-Yang \cite{LHY} considered a Bose gas at density $\rho > 0$, interacting through a two-body potential with scattering length $\frak{a} > 0$, and predicted that its ground state energy per unit volume has the form 
\begin{equation}\label{eq:LHY} 
e (\rho) = 4\pi \frak{a} \rho^2 \Big( 1+ \frac{128}{15 \sqrt{\pi}} (\rho \frak{a}^3)^{1/2} + \dots \Big) \end{equation} 
up to lower order terms, in the dilute limit $\rho \frak{a}^3 \ll 1$.

The first step towards a mathematical understanding of (\ref{eq:LHY}) is due by Dyson, who obtained a rigorous upper bound for $e(\rho)$, correct at leading order. The matching lower bound was derived by Lieb-Yngvason \cite{LY}. In the last years, the focus shifted to the second-order contribution on the r.h.s. of (\ref{eq:LHY}), known as the Lee-Huang-Yang term. In the so-called Gross-Pitaevskii regime, second order estimates for the ground state energy and for the low-energy excitation spectrum have been established in \cite{BBCS1,BBCS2} and, more recently, in \cite{HST,B}. In the thermodynamic limit, a lower bound capturing the Lee-Huang-Yang term has been shown for repulsive, quickly decaying and radial interactions, first in \cite{FS1}, under the additional assumption that the potential is integrable, and then in \cite{FS2}, in the general case. The first rigorous upper bound resolving the Lee-Huang-Yang term was proved in \cite{YY}, for sufficiently regular potentials. A simpler upper bound was later derived in \cite{BCS} for a larger class of repulsive, compactly supported and radial interactions; also this result, however, required $V \in L^3 (\bR^3)$ and excluded potentials with a hard core. For $V \in L^2 (\bR^3)$, a more precise upper bound, resolving also the third order term in the asymptotic expansion of $e(\rho)$, first predicted in the physics literature by Wu \cite{Wu}, Hugenholtz-Pines \cite{HP} and Sawada \cite{Sa}, was recently established in \cite{BOSS}. The corresponding lower bound is still open (but has been proven in the Gross-Pitaevskii regime in \cite{COSS}). New approaches have been developed in \cite{HHNST,HHST,FJGMOT} to extend (\ref{eq:LHY}) to an estimate for the free energy density, at low but positive temperatures.

Very recently, the restriction to integrable interactions appearing in \cite{YY,BCS} was lifted in \cite{BBCOS}, where an upper bound matching the Lee-Huang-Yang formula (\ref{eq:LHY}) was shown for interacting hard spheres. The goal of the present paper is to extend the analysis of \cite{BBCOS}, to derive an upper bound at Lee-Huang-Yang's order for general, repulsive, short-range interactions. 

We are going to consider a system of $N$ bosons moving in the box $\Lambda = [-L/2 ; L/2]^3$, with periodic boundary conditions and interacting through a repulsive potential $V : \mathbb{R}^3 \longrightarrow [0,\infty]$. We assume $V$ to be measurable, even and compactly supported, ie. $V(-x) = V(x)$ for all $x \in \bR^3$ and there exists $r_0 > 0$ with $\text{supp } V \subset B_{r_0} (0)$. Notice that our assumptions allow for potentials with a hard-core, ie. with $V(x) = \infty$, for all $|x| \leq r$, for some $r > 0$. The Hamilton operator has the form 
\begin{equation}\label{eq:HN0} H_N = \sum_{j=1}^N -\Delta_{x_j} + \sum^N_{i<j} V (x_i - x_j) \end{equation}
For a precise definition of $H_N$, we follow \cite{FS2} and introduce the quadratic form
\begin{align*}
    \mathcal{Q}_N [\Psi]:=\int_{\Lambda^N} \Big(\sum_{j=1}^N |\nabla_{x_j}\Psi|^2+\sum_{i<j}^N V(x_i-x_j)|\Psi|^2\Big) dx_1 \dots dx_N,
\end{align*}
acting on the domain 
\[ D(Q_N) = \Big\{ \Psi\in H_s^1(\Lambda^N) : \big(\sum_{i<j}^N V(x_i-x_j)\big)^{1/2}\Psi \in L^2(\Lambda^N) \Big\} \]
where $H_s^1 (\Lambda^N)$ is the permutation symmetric subspace of $H^1 (\Lambda^N)$. On the Hilbert space 
\[ \cH_{\mathcal{Q}_N}:=\overline{ \{ \Psi\in L^2_s (\Lambda^N): \mathcal{Q}_N [\Psi]<\infty\}},\] the non-negative and closable form $\mathcal{Q}_N$ is densely defined; therefore the self-adjoint operator $H_N$ can be defined as its Friedrichs extension.

To introduce the scattering length of $V$, we consider the energy functional 
\begin{equation}\label{eq:Ef}
    \mathcal{E}[g]:=\int \left(2|\nabla g (x)|^2+V(x) |g (x)|^2 \right)dx .
\end{equation}
We denote by $f$ the unique minimizer of (\ref{eq:Ef}), on the space of all real-valued $g$, such that $g-1\in \dot{H}_1(\mathbb R^3)$; see \cite{RT}. The scattering length $\frak{a} \geq 0$ of $V$ is defined by 
\begin{equation}\label{eq:defa}
    8\pi \mathfrak{a}:= \mathcal{E} [f] = \inf_g \mathcal{E}[g].
\end{equation} 

We denote by $E_{N,L}$ the ground state energy of (\ref{eq:HN0}). We are interested in the ground state energy per unit volume, in the thermodynamic limit $N, L \to \infty$, at fixed density $\rho = N / L^3$, which is defined by 
\begin{equation}\label{eq:erho} 
e (\rho) = \lim_{N,L \to \infty : N/ L^3 = \rho} \frac{E_{N,L}}{L^3}  \, . 
\end{equation} 
The existence of the limit (\ref{eq:erho}) is well-known, see \cite{R}. Our main theorem provides an upper bound for (\ref{eq:erho}), in the dilute limit. 
\begin{theorem}\label{thm:main}
Let $V : \bR^3 \to [0;\infty]$ be measurable, even, with $\text{supp } V \subset B_{r_0} (0)$, for some $r_0 > 0$. Let $\frak{a} > 0$ be the scattering length of $V$, defined as in (\ref{eq:defa}). Moreover, let $e(\rho)$ denote the ground state energy density of the interacting Bose gas, defined in (\ref{eq:erho}). Then there exist $C > 0$ and a sufficiently small $\delta > 0$ such that 
\begin{equation}\label{eq:main} e(\rho) \leq 4\pi \frak{a} \rho^2 \Big( 1+ \frac{128}{15 \sqrt{\pi}} (\rho \frak{a}^3)^{1/2} + C (\rho \frak{a}^3)^{1/2+\delta} \Big) \end{equation} 
for all $\rho \frak{a}^3 > 0$ small enough. 
\end{theorem}

Our bound (\ref{eq:main}) is consistent with the Lee-Huang-Yang formula (\ref{eq:LHY}). The corresponding lower bound was established in \cite{FS2}, for measurable, spherical symmetric potentials $V : \bR^3 \to [0;\infty]$ decaying sufficiently fast at infinity. Combining Theorem \ref{thm:main} with the results of \cite{FS2}, we obtain a rigorous proof of the Lee-Huang-Yang formula, for all spherical symmetric, repulsive interactions with compact support.

\medskip

{\it Acknowledgements.} G.B., S.C., and A.O. acknowledge funding from the European Research Council under the ERC Starting Grant MaTCh (Grant Agreement No. 101117299). G.B., S.C., and A.O. also thank GNFM–INdAM.
A.O. acknowledges financial support from the MUR Grant “Dipartimento di Eccellenza 2023-2027” of Dipartimento di Matematica, Politecnico di Milano, and from the Politecnico di Milano ``Seed Fund Grant''. M.B. acknowledges funding from the UZH Postdoc Grant FK-25-103.

\section{The trial state}

To prove Theorem \ref{thm:main}, we are going to construct an appropriate trial state. To this end, it is convenient to work on the bosonic Fock space
\[ \cF = \bigoplus_{n=0}^\infty L^2 (\Lambda)^{\otimes_s n} \, . \]
On $\cF$, we introduce the usual creation and annihilation operators $a^* (f), a(f)$, defined for an arbitrary $f \in L^2 (\Lambda)$ by 
\[ \begin{split} (a (f) \Psi)^{(n)}  (x_1, \dots , x_n) &= \sqrt{n+1} \int_{\Lambda}  dx \bar{f} (x) \Psi^{(n+1)} (x, x_1, \dots , x_n) \\ 
(a^* (f) \Psi)^{(n)} (x_1, \dots, x_n) &= \frac{1}{\sqrt{n}} \sum_{j=1}^n f (x_j) \Psi^{(n-1)} (x_1, \dots , x_{j-1}, x_{j+1}, \dots , x_n)  \end{split} \]
and satisfying the canonical commutation relations 
\begin{equation*}\label{eq:CCR} 
\big[ a (f) , a^* (g) \big] = \langle f , g \rangle , \qquad \big[ a (f) , a(g) \big] = \big[ a^* (f) , a^* (g) \big] = 0  \,.
\end{equation*} 
For $p \in \Lambda^* = 2\pi \bZ^3/L$, we define the momentum-space creation and annihilation operators 
\[ \hat{a}_p^* = a^* (f_p), \quad \hat{a}_p = a (f_p) \]
with the normalized plane wave $f_p (x) = e^{ip\cdot x} / L^{3/2}$. In position space, we will use  operator-valued distributions $a_x^* , a_x$, defined for $x \in \Lambda$, so that 
\[ a^* (f) = \int f (x) a^*_x \, dx , \qquad a (f) = \int \bar{f} (x) \, a_x dx \]
for every $f \in L^2 (\Lambda)$. 

On $\cF$, we define the Hamilton operator $\cH = \cK + \cV$, where the kinetic energy operator $\cK$ can be written in momentum space as 
\begin{equation}\label{eq:cK} \cK = \sum_{p \in \Lambda^*} p^2 \hat{a}_p^* \hat{a}_p \end{equation}
while the potential energy operator $\cV$ is conveniently expressed in position space as 
\begin{equation}\label{eq:cV} \cV = \frac{1}{2} \int dx dy \, V(x-y) a_x^* a_y^* a_y a_x  \,. \end{equation} 
Our goal is to construct a normalized trial state $\Psi \in \cF$, with approximately $N = \rho L^3$ particles and with energy per unit volume $\langle \Psi, \cH \Psi \rangle / L^3$ sufficiently close to the r.h.s. of (\ref{eq:main}), which will then follow by a standard equivalence of ensembles argument. To define such a trial state, we have to generate, first of all, a Bose-Einstein condensate in the zero-momentum state, with density $0 < \rho_0 \leq \rho$. To reach this goal, we will use the unitary Weyl operator 
\begin{equation}\label{eq:weyl} W(\rho_0) = e^{a^* (\sqrt{\rho_0}) - a (\sqrt{\rho_0})} = e^{\sqrt{\rho_0} \int dx (a_x^* - a_x)} \end{equation} 
producing the coherent state 
\begin{equation} \label{eq:coh} W(\rho_0) \Omega = e^{-\rho_0/2} \Big\{ 1, \sqrt{\rho_0}, \dots, \frac{\rho^{k/2}_0}{\sqrt{k!}}, \dots \Big\} \in \cF \,.\end{equation}
Energetically, choosing $\rho_0 = \rho$, the pure condensate (\ref{eq:coh}) is very far from the true ground state energy, because it completely lacks correlations. In fact, if the potential $V$ has a hard-core, the energy of (\ref{eq:coh}) is even infinite. Hence, we have to modify (\ref{eq:weyl}), to introduce the correct correlation structure, making sure, in particular, that the trial state vanishes on hard-cores. To this end we proceed as in \cite{BBCOS}. We introduce two length scales $\ell = \frak{a} (\rho \frak{a}^3)^{-\delta}$ and $\ell_0 = (\rho \frak{a})^{-1/2} (\rho \frak{a}^3)^{-\eps}$, for $\delta , \eps > 0$ small enough, to be fixed later on. In the dilute limit $\rho \frak{a}^3 \ll 1$, $\ell$ is a bit larger than the range of the interaction while $\ell_0$ is just a bit larger than the healing length $(\rho \frak{a})^{-1/2}$ of the gas, ie. $\frak{a} \leq r_0 \ll \ell \ll (\rho \frak{a})^{-1/2} \ll \ell_0$. We will modify the coherent state (\ref{eq:coh}) adding a Jastrow factor acting on short length scales $|x| \leq \ell$ and a Bogoliubov transformation generating correlations on larger scales, $\ell \leq |x| \leq \ell_0$. The Jastrow factor is crucial for the trial state to vanish on hard-cores. The Bogoliubov transformation, on the other hand, is crucial to describe correlations up to (and just beyond) the healing length, as needed to reach Lee-Huang-Yang precision. 

To define the Jastrow factor on short length scales, we need some properties of the minimizer of the energy functional (\ref{eq:Ef}), defining the scattering length of the interaction potential $V$.
\begin{lemma}
\label{Lem:general_potential_eff}
Assume that $V : \bR^3 \to [0;\infty]$ is measurable, even, with $\text{supp } V \subset B_{r_0} (0)$, for some $r_0 > 0$. Let $f$ denote the unique minimizer of (\ref{eq:Ef}), on the space of all $f$ such that $f-1\in \dot{H}_1(\mathbb R^3)$. We consider the distribution 
\begin{equation}\label{eq:Veff} V_\mathrm{eff} = 2\Delta f \, . \end{equation}
Then, $V_\mathrm{eff}$ is non-negative, with support in $B_{r_0}(0)$. Its Fourier transform $\widehat{V}_\mathrm{eff} \in L^\infty (\bR^3)$ is such that $\widehat{V}_\mathrm{eff} (0) = 8\pi \frak{a}$ and, for every $m \in \bN$, $\| \nabla^m \widehat{V}_\mathrm{eff} \|_\infty < \infty$. Furthermore, for $|x|> r_0$, we have 
    \begin{align}
    \label{Eq:formula_f_integ}
        f(x)=1-\int \frac{V_\mathrm{eff}(y)}{8\pi |x-y|} dy\,.
    \end{align}
    In particular, it follows that 
    \begin{equation} \label{eq:1-f} 
    0 \leq 1- f (x) \lesssim \frac{\frak{a}}{|x|} \end{equation} 
    for all $x \in \bR^3$ and 
    \begin{equation}\label{eq:nablaf} |\nabla f (x)| \lesssim \frac{\frak{a}}{x^2} \end{equation} 
for all $|x| > 2 r_0$. \end{lemma}

{\it Remark.} Due to the compact support of $V_\mathrm{eff}$, both the Fourier transform
\begin{align*}
   \widehat{V}_\mathrm{eff}(k) = \int e^{-iky}V_\mathrm{eff}(y) dy := \int e^{-iky} \varphi(y)V_\mathrm{eff}(y) dx
\end{align*}
as well as the expression 
\begin{align}
\label{Eq:extension_to_bounded_func}
    \int \frac{V_\mathrm{eff}(y)}{8\pi |x-y|} dy :=  \int \frac{\varphi(y)V_\mathrm{eff}(y)}{8\pi |x-y|} dy
\end{align}
are well-defined, for $|x| > r_0$, if we choose $\varphi$ a smooth and compactly supported test function with $\varphi=1$ on $B_{r'}(0)$ and $\varphi=0$ on $\mathbb R^3\setminus B_{r''}(0)$ for some $r_0 < r' < r'' < |x|$. Notably, these definitions are independent of the concrete choice of $\varphi$.

{\it Remark.} If we additionally assumed that $V$ is radial, we would find $f(x) = 1 - \frak{a} / |x|$, for $|x| > r_0$. This would immediately imply the bounds (\ref{eq:1-f}), (\ref{eq:nablaf}). 

\begin{proof}[Proof of Lemma \ref{Lem:general_potential_eff}]
From the variational definition, using test functions of the form $f\pm \epsilon g$ and sending $\epsilon\rightarrow 0$, it is immediately clear that $V_\mathrm{eff}$ is supported on $B_{r_0}(0)$, i.e.
\begin{align}
\label{Eq:Scattering_I}
    \langle V_\mathrm{eff}, g\rangle=0,
\end{align}
if $g(x)=0$ for all $|x|\leq r_0$. Similarly, using test functions of the form $f+ \epsilon g$, $\epsilon > 0$, we conclude that $V_\mathrm{eff}$ is non-negative, i.e.
\begin{align*}
     \langle V_\mathrm{eff}, g\rangle\geq 0
\end{align*}
in case $g\geq 0$. To prove (\ref{Eq:formula_f_integ}), let $m_\epsilon(y):=\epsilon^{-3}m(\epsilon^{-1} y)$, for $\epsilon > 0$ and for a smooth function $m$, supported in $B_1 (0)$ and with $\int m(y) \mathrm{d}y=1$. For every $\epsilon > 0$, $m_\epsilon * V_\text{eff}$ is a smooth function, supported in $B_{r_0+\epsilon} (0)$. Choosing $\varphi$ as in \eqref{Eq:extension_to_bounded_func} and $\epsilon > 0$ so small that $r_0 + \epsilon < r'$, we conclude that 
\begin{align*}
    \int \frac{V_\mathrm{eff}(y)}{8\pi |x-y|} dy=\lim_{\epsilon\rightarrow 0}\int \frac{(m_\epsilon* V_\mathrm{eff})(y)}{8\pi |x-y|} dy \,.
\end{align*}
Consequently, to verify \eqref{Eq:formula_f_integ}, it is enough to show that, for all $\epsilon>0$ and $x \in \bR^3$,
\begin{align}
\label{Eq:formula_f_integ_in_proof}
    m_\epsilon*(f-1) (x) =-\int \frac{(m_\epsilon* V_\mathrm{eff})(y)}{8\pi |x-y|} dy \,.
\end{align}
By the definition of $V_\mathrm{eff}$, it is clear that $2\Delta m_\epsilon*(f-1) = m_\epsilon * V_\mathrm{eff}$. Furthermore, the smooth function on the right hand side of \eqref{Eq:formula_f_integ_in_proof} satisfies 
\begin{align*}
    2\Delta \left(-\int \frac{(m_\epsilon* V_\mathrm{eff})(y)}{8\pi |x-y|} dy\right)= (m_\epsilon * V_\mathrm{eff}) (x)
\end{align*}
as well. This concludes the proof of \eqref{Eq:formula_f_integ_in_proof}, as both the function on the left hand side as well as the one on the right hand side of \eqref{Eq:formula_f_integ_in_proof} vanish at infinity.  

In order to verify that $\widehat{V}_\text{eff} (0) = 8\pi \frak{a}$ we observe, first of all, that 
\begin{equation}\label{eq:Veff0} \widehat{V}_\text{eff} (0) = \lim_{R \to \infty} 2 \int_{B_R (0)} \Delta f \, dx = \lim_{R \to \infty} 2 \int_{\partial B_R (0)} \nabla f  \cdot n \, d\sigma \, . \end{equation} 
On the other hand, using the variational definition of $\frak{a}$ and observing that $f \pm \epsilon \varphi f$ is an admissible test function, for a smooth and compactly supported $\ph$, with $\ph = 1$ on the support of $V$, we obtain 
\begin{align*}
    0&=\int \left[2 \nabla f \cdot  \nabla (\varphi f)+V f (\varphi f)\right] dx \\ &=\int \left[2 \nabla f \cdot  \nabla (\varphi f)+V f^2\right] dx \, = 8\pi \frak{a} + \int 2 \, \nabla f \cdot \nabla (\ph f - f)  \, dx 
\end{align*}
implying 
\[ \int 2 \, \nabla f \cdot \nabla (f-\ph f) \, dx = 8\pi \frak{a} \, . \]
Since furthermore $-\Delta f=0$ on the support of $1-\varphi$, we conclude that 
\[ \begin{split} \int 2 \, \nabla f \cdot \nabla (f-\ph f) \, dx &= \lim_{R \to \infty} \int_{B_R (0)} 2 \, \nabla f \cdot \nabla (f-\ph f) \, dx \\ &= \lim_{R \to \infty}
\int_{\partial B_R (0)} 2\nabla f \cdot n  \, (f - \ph f)  \, d\sigma \\ &= \widehat{V}_\text{eff} (0) + \lim_{R \to \infty} \int_{\partial B_R (0)}  2\nabla f \cdot n  \, (f-1) \, d\sigma \end{split} \]
where in the last line we used (\ref{eq:Veff0}) and the fact that $\ph =0$ on $\partial B_R (0)$, for $R$ large enough. To show that $\widehat{V}_\text{eff} (0) = 8\pi \frak{a}$, we just observe that, from \eqref{Eq:formula_f_integ}, $|f (x)-1| \lesssim 1/R$, $|\nabla f (x)| \lesssim 1/R^2$ for $x \in \partial B_R (0)$ and $R$ large enough, implying that the integral on the r.h.s. vanishes, as $R \to \infty$. From $\widehat{V}_\text{eff} (0) = 8\pi \frak{a}$ we also obtain, since $V_\mathrm{eff} \geq 0$, that $\| V_\mathrm{eff} \|_1 = \sup \langle V_\mathrm{eff} , g \rangle = 8\pi \frak{a}$, where the supremum is taken over all test functions $g$, with $\| g \|_\infty \leq 1$. Similarly, as a consequence of \eqref{Eq:Scattering_I}, the $L^1$-norm of $x_{i_1}\dots x_{i_m}V_\mathrm{eff}(x)$ is finite as well, and therefore 
\begin{align*}
    \|\partial_{k_{i_1}}\dots \partial_{k_{i_m}} \widehat{V}_\mathrm{eff}\|_\infty < \infty \,.
\end{align*}

Finally, to show (\ref{eq:1-f}), (\ref{eq:nablaf}), we use \eqref{Eq:formula_f_integ} and we remark that, for $|x| > 2 r_0$ and $|y| \leq r_0$, $|x-y| \geq |x|/2$. Together with the identity $\widehat{V}_\text{eff} (0) = 8\pi \frak{a}$, we immediately find the desired bounds, for $|x| \geq 2 r_0$. Since $0 \leq 1- f(x) \leq 1$ everywhere (by the variational characterization of $f$), (\ref{eq:1-f}) can be extended to all $x \in \bR^3$.  
\end{proof}

Next, we truncate the minimizer $f$ of the energy functional (\ref{eq:Ef}) at the length scale $\ell = \frak{a} (\rho \frak{a}^3)^{-\delta}$, defining 
\begin{equation}\label{eq:fell} f_\ell (x) = \chi_\ell (x) f(x) + (1- \chi_\ell (x)) = 1- \chi_\ell (x) (1 -f(x)) \end{equation}
where $\chi_\ell (x) = \chi (x/\ell)$ and $\chi \in C^\infty_c (\bR^3)$, with $\chi (x) = 1$ for $|x| \leq 2$, $\chi (x) = 0$, for $|x| > 4$. 
\begin{lemma}  \label{lm:omega}
We have $0\leq f_\ell (x) \leq 1$ for all $x \in \bR^3$. Let $\omega_\ell (x) = 1- f_\ell (x)$. Then 
\[ \omega_\ell (x) \lesssim \frac{\frak{a}}{|x|} \chi_\ell (x) \]
for all $x \in \bR^3$ and 
\[ |\nabla f_\ell (x)| = |\nabla \omega_\ell (x)| \lesssim \frac{\frak{a}}{x^2} \]
for all $|x| > 2 r_0$. Moreover, 
\begin{equation}\label{eq:en-fl} \Big| \int \big[ |\nabla f_\ell (x)|^2 + \frac{1}{2} V (x) |f_\ell (x)|^2 \big] dx - 4 \pi \frak{a} \Big| \lesssim \frac{\frak{a}^2}{\ell} \,.\end{equation}
\end{lemma} 
\begin{proof}
The bounds for $f_\ell, \omega_\ell, \nabla f_\ell$ follow from Lemma \ref{Lem:general_potential_eff}. To show (\ref{eq:en-fl}), we recall the definition (\ref{eq:defa}) and we observe that, since $f_\ell (x) = f(x)$ on the support of $V$ and since $\nabla f_\ell (x) = \chi_\ell (x)\nabla f (x)  - (1-f(x)) \nabla \chi_\ell (x)$, 
\[ \begin{split} \Big| \int  \big[ |\nabla f_\ell (x)|^2 &+ \frac{1}{2} V (x) | f_\ell (x)|^2\big] dx  - 4\pi \frak{a} \Big| \\ \leq \; &\int |\nabla f (x)|^2 (1- \chi_\ell^2 (x)) dx + \int |\nabla \chi_\ell (x)|^2 (1- f(x))^2 \\ &+2 \int |\nabla f (x)| |\nabla \chi_\ell (x)| (1- f(x)) dx \,. \end{split} \]
From Lemma \ref{Lem:general_potential_eff}, we have $|1-f(x)| \lesssim \frak{a}/|x|$ and $|\nabla f (x)| \lesssim \frak{a}/x^2$ on the supports of $1-\chi_\ell^2$ and of $\nabla \chi_\ell$. With $|\nabla \chi_\ell (x)| \lesssim \mathbbm{1} (2\ell \leq |x| \leq 4\ell) / \ell$, we obtain (\ref{eq:en-fl}). 
\end{proof}

With $f_\ell$ as in (\ref{eq:fell}), we define $J : \cF \to \cF$, requiring that  
\begin{equation}\label{eq:defJ} (J \Phi)^{(n)} (x_1, \dots , x_n) = \prod_{i<j}^n f_\ell (x_i - x_j) \Phi^{(n)} (x_1, \dots , x_n) \,. \end{equation}
For $x \in \Lambda$, we also define $J(x) : \cF \to \cF$ by 
\[ (J(x) \Phi)^{(n)} (x_1, \dots , x_n) = \prod_{j=1}^n f_\ell (x - x_j) \Phi^{(n)} (x_1, \dots , x_n) \,.\]
Then, we have the identities 
\[ a_x J = J(x) J a_x , \qquad J a_x^* = a_x^* J(x) J \]
and
\[ a_y J(x) = f_\ell (x-y) J(x) a_y, \qquad J(x) a_y^* = f_\ell (x-y) a_y^* J(x) \,.\]
Clearly, $0 \leq J(x), J \leq 1$ for all $x \in \bR^3$. From
\[ 1 - \prod_{j=1}^n f_\ell (x- x_j) \leq \sum_{j=1}^n (1 - f_\ell (x-x_j)) = \sum_{j=1}^n \omega_\ell (x-x_j) \]
we also find 
\begin{equation}\label{eq:1-J} 1 - J(x) \leq d\Gamma (\omega_{\ell,x}) = \int dy \, \omega_\ell (x-y) a_y^* a_y \,. \end{equation}

Acting with (\ref{eq:defJ}) on the condensate (\ref{eq:coh}) we generate a Jastrow factor, capturing correlations up to the length scale $\ell$ and making sure that the trial state vanishes on hard-cores (if $V (x) = \infty$ for all $|x| \leq r$, then clearly the minimizer $f$ of (\ref{eq:Ef}) satisfies $f(x) = 0$ for all $|x| \leq r$). To create correlations on larger length scales, we are going to use a Bogoliubov transformation. Proceeding as in \cite{BBCOS}, we define the coefficients \begin{equation}\label{eq:hatsk} \begin{split} 
\widehat{s}_p &= - \frac{p^2 + \rho_0 \widehat{V}_\text{eff} (p) -  \sqrt{|p|^4 + 2p^2 \rho_0 \widehat{V}_\text{eff} (p)}}{\sqrt{\rho_0^2 \widehat{V}_\text{eff} (p)^2 - \Big( p^2 + \rho_0 \widehat{V}_\text{eff} (p) -  \sqrt{|p|^4 + 2p^2 \rho_0 \widehat{V}_\text{eff} (p)}\Big)^2}}
\end{split} 
\end{equation}
for every $p \in \Lambda^* \backslash \{ 0 \}$, with $V_\mathrm{eff}$ as in (\ref{eq:Veff}). We set also $\widehat{s} (0) = 0$ and we denote by $s(x)$ the periodic function on $\Lambda$, with Fourier coefficients (\ref{eq:hatsk}). Then, we define  
\begin{equation}\label{eq:tlsigma} \tilde{\sigma} (x) = \frac{\chi_{\ell_0} (2x) (1- \chi_\ell (2x))}{f_\ell (x)} \big[ \rho_0 \omega_\ell (x) + s (x) \big] \end{equation} 
and 
\begin{equation}
\label{eq:def-sigma} \sigma (x) = \tilde{\sigma} (x) - \ell_0^{-3} \Big[ \int \tilde{\sigma} (z) dz \Big]  \ph (x/\ell_0) 
\end{equation} 
so that $\widehat{\sigma}_0 = 0$. We set now $\widehat{\eta}_p = \mathrm{sinh}^{-1} (\widehat{\sigma}_p)$, with $\widehat{\sigma}_p$ denoting the Fourier coefficients of (\ref{eq:def-sigma}), and we consider the Bogoliubov transformation 
\begin{equation}\label{eq:def-bog} T = \exp \Big[ \frac{1}{2} \sum_{p \in \Lambda^* \backslash \{ 0 \}} \widehat{\eta}_p \big( \hat{a}_p^* \hat{a}_{-p}^* - \hat{a}_p \hat{a}_{-p} \big) \Big] \,. \end{equation} 
The action of $T$ on creation and annihilation operators is explicitly given by 
\[ T^* \hat{a}_p T = \widehat{\gamma}_p \hat{a}_p + \widehat{\sigma}_p \hat{a}_{-p}^* \]
for every $p \in \Lambda^*$, with $\widehat{\gamma}_p = \cosh \eta_p = \cosh (\sinh^{-1} (\widehat{\sigma}_p))$. In the next lemma, which is an adaptation of \cite[Lemma 2.2]{BBCOS}, we collect some important properties of the kernels $\eta, \sigma, \gamma$. We briefly discuss its proof (focusing, in particular, on the differences w.r.t. the proof of \cite[Lemma 2.2]{BBCOS}) in Appendix \ref{app:kernels}. 
\begin{lemma} \label{lm:eta}
The coefficients $\widehat{s}_k$, defined in \eqref{eq:hatsk}, satisfy 
\begin{equation} \label{eq:sfourier} |\widehat{s}_k| \leq C \min \Big\{  \frac{\rho^{1/4}}{|k|^{1/2}} , \frac{\rho}{k^2} \Big\} \end{equation} 
for all $k \in \Lambda^*_+ = \Lambda^* \backslash \{ 0 \}$. For the function $s : \Lambda \to \bR$ with Fourier coefficients $\widehat{s}_k$ we find the pointwise bounds 
\begin{equation*} | s (x)| \leq C \frac{\rho}{|x|} \min \Big\{ 1 , \frac{1}{(\rho^{1/2} |x|)^{3/2}} \Big\} , \qquad |\nabla s (x)| \leq C \frac{\rho |\hspace{-.05cm} \log \rho |}{x^2} \end{equation*} 
for all $|x| \geq 15r_0$. Moreover 
\begin{equation}\label{eq:est-combi}
 \int_{|x|\geq \ell_0} |s(x)|^2 dx  \leq C\ell_0^{-2}\rho^{\frac{1}{2}}  , \;   \int_{2\frak{a} \leq |x|\leq  2\ell} 
 |s(x)|^2 dx  \leq C\ell\rho^{2}, \;  \int_{|x|\geq \ell_0} |\nabla s(x)|^2 dx  \leq  C\ell_0^{-4}\rho^{\frac{1}{2}}\, . \end{equation} 
 %
 %
%
Let now $\zeta\in \{\tilde\sigma, \sigma,\gamma-\mathbbm{1},\mathbbm{1} - \gamma^{-1} , \gamma^{-1}*\sigma\}$. Then we have 
\begin{equation} \label{eq:zeta-combi1} 
 \|\zeta\|^2_2 \leq C \rho^{\frac{3}{2}}, \qquad \| \zeta \|_\infty \leq C\, \frac{\rho}{\ell} , \qquad  \|\zeta\|_1  \leq C \big(\rho^{\frac{1}{2}}\ell_0\big)^3 
 \end{equation} 
 and also 
 \begin{equation}\label{eq:zeta-combi2}
   \|\nabla \zeta\|^2_2 \leq C \rho^{2} , \qquad  \|\nabla \zeta\|_\io  \leq C \, \frac{\rho}{\ell^2}
   \end{equation} 
%
and, for any $m \in \bN$, the pointwise decay estimate 
\begin{equation}
            \label{eq:decay}
 |\zeta(x)|  \leq C \, \frac{\rho}{\ell}\Big( \frac{\rho^{\frac{1}{4}}\ell_0^{\frac{3}{2}}}{|x|}\Big)^m \leq C \, \frac{\rho}{\ell} \, \frac{1}{\big( \rho^{1/2+3\eps/2} |x| \big)^{m}} \,.
\end{equation}
For $\zeta \in \{ \tilde{\sigma}, \sigma \}$, we get stronger estimates. In particular
\begin{equation}\label{eq:s2-restricted}
\int_{|x|\geq \ell_0} \big[ |\tilde{\sigma} (x)|^2 + |\sigma (x)|^2 \big] dx  \leq C\ell_0^{-2}\rho^{\frac{1}{2}}  , \qquad  \int_{|x|\leq  2\ell} 
 \big[ |\tilde{\sigma} (x)|^2 + | \sigma (x)|^2 \big] dx \leq C\ell\rho^{2}
 \end{equation} 
and  
\begin{equation} 
            \label{Th:est_3}
             \|\sigma\|_1  \leq 2\|\tilde \sigma\|_1\leq C \big( \rho^{\frac{1}{2}}\ell_0\big)^\frac{1}{2}\,.
             \end{equation}
Additionally, we get the pointwise bounds 
\begin{equation}\label{eq:snabla-point} |\nabla \tilde{\sigma} (x)| , |\nabla \sigma (x) | \leq C \frac{\rho | \hspace{-.03cm} \log \rho|}{x^2} \end{equation} 
for all $x \in \Lambda$, which imply  
\begin{equation}\label{eq:snabla-p}\| \nabla \sigma \|_p, \| \nabla \tilde{\sigma} \|_p  \leq C \rho |\hspace{-.03cm} \log \rho  | \left\{ \begin{array}{ll}  \ell^{3/p-2} \quad \text{if } p > 3/2 \\ \ell_0^{3/p-2} \quad \text{if } p < 3/2 \end{array} \right.\end{equation} 
and
\begin{equation}\label{eq:nu-nabla-p} \| \nabla (\gamma^{-1}*\sigma) \|_p  \leq C (\rho^{1/2} \ell_0 )^3 \rho |\hspace{-.03cm} \log \rho | \left\{ \begin{array}{ll}  \ell^{3/p-2} \quad &\text{if } p > 3/2 \\  \ell_0^{3/p-2} \quad &\text{if } p < 3/2 \end{array} \right.\end{equation} 
for all $1 \leq p \leq \infty$.
\end{lemma}

We are going to consider the trial state 
\begin{equation}\label{eq:trial} \Psi = \frac{1}{Z} J W(\rho_0) T \Omega \in \cF \end{equation}
with $J$ as in (\ref{eq:defJ}), the Weyl operator $W(\rho_0)$ as in (\ref{eq:weyl}), the Bogoliubov transformation $T$ from (\ref{eq:def-bog}) and with the normalization constant $Z = \| J W(\rho_0) T \Omega \|$. We choose $0 < \rho_0 < \rho$ such that 
\begin{equation}\label{eq:choice} \rho+ C _1 \rho^{7/4-11\eps-\delta} \leq \rho_0 + \| \sigma \|^2 \leq \rho + C_2 \rho^{7/4-11\eps-\delta} \end{equation} 
for fixed constants $C_1 < C_2$. This is possible because  $\| \sigma \|_2^2 \simeq \rho^{3/2} \gg \rho^{7/4-11\eps-\delta}$.

As observed in \cite[Lemma 2.3]{BBCOS}, we have the important identity 
\begin{equation}\label{eq:id} a_x \Psi = \sqrt{\rho_0}J(x)\Psi+J(x)\int dy\,(\gamma^{-1} * \sigma)(x-y)a^*_y J(y)\Psi.
\end{equation}

The proof of Theorem \ref{thm:main} is based on the following two propositions, estimating number of particles and energy of the trial state (\ref{eq:trial}). 
\begin{proposition} \label{prop:N_on_psi}
	We have 
	\begin{equation} \label{eq:cN_ub}
		\big| \langle \Psi,\mathcal{N}\Psi\rangle - \rho L^3 \big| \leq C \rho^{7/4-11\eps-\delta} L^3
	\end{equation}    
	if $\eps , \delta > 0$ are small enough.  	
\end{proposition}

\begin{proposition} \label{prop:Psi-energy} 
Let $0 < \delta < \eps/2$, with $\eps > 0$ small enough. Then there exist $C>0$ such that
	\begin{equation} \label{eq:Hpsi0}
		L^{-3} \langle \Psi,\mathcal{H} \Psi\rangle \leq 4 \pi \frak{a} \rho^2 \Big(1 + \frac{128}{15\sqrt{\pi}} (\rho \frak{a}^3)^{1/2} + C (\rho \frak{a}^3)^{1/2+\delta} \Big) 
	\end{equation}
	for all $\rho \frak{a}^3 > 0$ small enough. 
\end{proposition}

\begin{proof}[Proof of Theorem \ref{thm:main}]
Given Prop. \ref{prop:N_on_psi} and Prop. \ref{prop:Psi-energy}, the proof of Theorem \ref{thm:main} is based on equivalence of ensembles. The argument is identical to the one used in \cite{BBCOS}. 
\end{proof}

\section{Number of particles in the trial state}

In this section, we show Prop. \ref{prop:N_on_psi}. To this end, we recall from \cite{BBCOS} that (\ref{eq:trial}) satisfies some important a-priori bounds on the local number of particles and also on the local number of excitations of the Bose-Einstein condensate, as measured by the fields 
\begin{equation}\label{eq:bb} \begin{split} b^*_x &= a^*_x - \sqrt{\rho_0} = W (\rho_0) a^*_x W(\rho_0)^* \\  b_x &= a_x - \sqrt{\rho_0} = W(\rho_0) a_x W(\rho_0)^*\,. \end{split} \end{equation} 
The following lemma is shown in \cite[Lemma 4.3, Lemma 4.4]{BBCOS}.
\begin{lemma} \label{lm:a-bds} 
Let $R = C \rho^{-1/2-3\eps}$, $k \in \bN \backslash \{ 0 \}$. Then we have 
\begin{equation}\label{eq:kas}  \int dx_1 \dots dx_k \, \prod_{j=2}^k \chi (|x_1 - x_j| \leq R)  \langle \Psi, a_{x_1}^* \dots a_{x_k}^* a_{x_k} \dots a_{x_1} \Psi \rangle \leq C \rho^{-(k-3)/2 - (12 k -9) \eps} L^3 \end{equation} 
and 
\begin{equation}\label{eq:kbs}  \int dx_1 \dots dx_k \, \prod_{j=2}^k \chi (|x_1 - x_j| \leq R)  \langle \Psi, b_{x_1}^* \dots b_{x_k}^* b_{x_k} \dots b_{x_1} \Psi \rangle \leq C \rho^{3/2 - (10k -9) \eps} L^3 \end{equation} 
if $\eps, \delta > 0$ and $\rho > 0$ are small enough (if $k=1$, we remove $\prod_{j=2}^k \chi (|x_1 - x_j| \leq R)$). Moreover, if we have only one annihilation (or creation) operator, we find \begin{equation}\label{eq:1bb} \Big| \sqrt{\rho_0} \int dx\langle \Psi, a_x \Psi \rangle - \rho_0 L^3 \Big| = \Big| \sqrt{\rho_0} \int dx \langle \Psi, b_x \Psi \rangle \Big| \lesssim \rho^{2-9\eps-2\delta} L^3 \end{equation} 
if $\eps, \delta > 0$ and $\rho >0$ are small enough. 
\end{lemma} 

Sometimes, it is useful to estimate the number of excitations generated only by the Jastrow factor (removing excitations created by the Bogoliubov transformation). Defining the new operators
\begin{equation}\label{eq:def-cc}
\begin{split} 
c^*_x &=b^*(\g_x)-b(\s_x)= a^* (\gamma_x) - a (\sigma_x) - \sqrt{\rho_0} = T W (\rho_0) a^*_x W^* (\rho_0) T^* \\
c_x &=b(\g_x)-b^*(\s_x)= a (\gamma_x) - a^* (\sigma_x) - \sqrt{\rho_0} =  T W (\rho_0) a_x W^* (\rho_0) T^* 
\end{split} \end{equation} 
we have the following bounds, established in \cite[Lemma 4.5]{BBCOS}.
\begin{lemma} \label{lm:c's}
We have 
\begin{equation}\label{eq:cc-claim} \int dx \langle \Psi, c_x^* c_x \Psi \rangle \leq C \rho^{2-16\eps-2\delta} L^3\,. \end{equation} 
Moreover, let $R = C \rho^{-1/2-3\eps}$. Then, we have 
\begin{equation}\label{eq:cccc-claim} \int_{|x-y| \leq R} dx dy  \, \langle \Psi,  c_x^* c_y^* c_y c_x  \Psi \rangle \leq C \rho^{2-30\eps-\delta} L^3 \,.
\end{equation} 
More generally, for $\ph, \tau \in \{ \sigma, \gamma^{-1}\ast \sigma, \gamma -\mathbbm{1}, \mathbbm{1} \}$, we find (if $\tau = \mathbbm{1}$, we set $c (\tau_x) = c_x$) 
\begin{equation}\label{eq:cc-gh}
\int dx \, \langle \Psi, c^* (\tau_x) c (\tau_x) \Psi \rangle \leq C \rho^{2-22\eps-2\delta} L^3 \end{equation}
and  
\begin{equation}\label{eq:cccc-gh} 
\int_{|x-y| \leq R} dx dy  \, \langle \Psi,  c^* (\tau_x) c^* (\ph_y) c (\ph_y) c (\tau_x)  \Psi \rangle \leq C \rho^{2-42\eps-\delta} L^3 \,.
\end{equation} 
\end{lemma} 
As an application of the last lemma, we can show stronger estimates for the local number of particles and of excitations, in small regions. The next lemma is taken from \cite[Lemma 4.6]{BBCOS} .
\begin{lemma} \label{lm:aaaaC} 
For $0 < \delta < \eps$, $\eps > 0$ small enough and $\rho > 0$ small enough, we have 
\begin{equation}\label{eq:aaaaC} \int_{|x-y| < C \ell} dx dy \, \langle \psi, b_x^* b_y^* b_y b_x \psi \rangle , \; \int_{|x-y| < C \ell} dx dy \, \langle \psi, a_x^* a_y^* a_y a_x \psi \rangle  \lesssim \rho^{2-36\eps-\delta} L^3\,. \end{equation} 
Moreover, recalling $R = C \rho^{-1/2-3\eps}$, we find  
\begin{equation}\label{eq:aaaCR} \int_{|x-y| \leq C \ell, |x-z| \leq R} dx dy dz \langle \Psi, a_x^* a_y^* a_z^* a_z a_y a_z \Psi \rangle \lesssim \rho^{3/2-45\eps-\delta} L^3\,. \end{equation}
\end{lemma} 

We can now prove Prop. \ref{prop:N_on_psi}, estimating the total number of particles in (\ref{eq:trial}). Also this proof is taken from \cite{BBCOS}. 
\begin{proof}[Proof of Prop. \ref{prop:N_on_psi}] 
Writing  
\begin{equation} \label{eq:a_to_b}
a_x = \sqrt{\rho_0} + b_x\,, \qquad a_x^* = \sqrt{\rho_0} + b_x^*
\end{equation} 
and 
\[ b_x = c (\gamma_x) + c^* (\sigma_x), \qquad b_x^* = c^* (\gamma_x) + c (\sigma_x) \]
we obtain 
\[ \begin{split}  \langle\Psi ,\mathcal{N} \Psi\rangle =\; &\rho_0 L^3 +  2\sqrt{\rho_0} \, \text{Re } \int dx \langle \Psi, b_x \Psi \rangle + \int dx \langle \Psi, b_x^* b_x \Psi \rangle 
\\ =\;&(\rho_0 + \| \sigma \|_2^2) L^3 +  2 \sqrt{\rho_0} \, \text{Re } \int dx \langle \Psi, b_x \Psi \rangle \\ &+ \int dx\, \langle\Psi, \Big[ c^*(\gamma_x) c(\gamma_x)+ c^*(\sigma_x) c(\sigma_x)+c^*(\gamma_x) c^*(\sigma_x)+c(\gamma_x) c(\sigma_x)\Big]\Psi\rangle. \end{split} \]
The claim now follows by (\ref{eq:choice}), combining Lemma \ref{lm:a-bds} and Lemma \ref{lm:c's}. 
\end{proof}

\section{Energy of the trial state} 

In this section, we prove Prop. \ref{prop:Psi-energy}. To this end, we are going to estimate separately the kinetic and the potential energy of the trial state (\ref{eq:trial}). The bound for the kinetic energy was shown in  \cite[Prop. 2.5]{BBCOS}.
\begin{proposition}\label{prop:kin}
Let $\Psi$ be defined as in (\ref{eq:trial}) and $\cK$ denote the kinetic energy operator (\ref{eq:cK}). Then, we have 
\begin{equation}\label{eq:Kpsi} \begin{split} \frac{1}{L^3} \langle \Psi, \cK \Psi \rangle \leq \; &\rho_0^2 \int |\nabla f_\ell(x)|^2 dx  + 2\rho_0 \|\s\|^2  \int |\nabla f_\ell(x)|^2 dx  + \int |\nabla f_\ell (x)|^2 |\sigma (x)|^2 dx \\ &+ 2\int f_\ell (x) \nabla f_\ell (x) \sigma (x) \nabla \sigma (x) dx  +  \int |f_\ell(x)|^2 |\nabla \s(x)|^2 dx \\	& + 2\rho_0 \int |\nabla f_\ell(x)|^2 \big[(\g \ast \s)(x) + (\s \ast \s)(x)\big] dx \\
& + 2\rho_0 \int  f_\ell(x) \nabla f_\ell(x) \cdot \nabla \s (x)  dx\, + C \rho^{5/2+\delta}
\end{split}\end{equation} 
for $0 < \delta < \eps$ with $\eps > 0$ small enough.
\end{proposition}

In the next proposition, whose proof is deferred to Sect. \ref{sec:pot}, we bound the potential energy of (\ref{eq:trial}). This is the main novelty of this paper. 
\begin{proposition}\label{prop:pot}
Let $\Psi$ be defined as in (\ref{eq:trial}) and $\cV$ denote the potential energy operator (\ref{eq:cV}). Then, we have 
\begin{equation}\label{eq:Vpsi}\begin{split} 
\frac{1}{L^3} \langle \Psi, \cV \Psi \rangle \leq \; &\frac{\rho_0^2}{2} \int dx \, V(x) f_\ell^2 (x) + \rho_0 \| \sigma \|_2^2 \int dx \, V(x) f_\ell^2 (x) \\ &+ \rho_0 \int dx \, V(x) f_\ell^2 (x)  \big( (\sigma * \sigma) (x) + (\sigma * (\gamma-1)) (x) \big) + C \rho^{5/2+ \eps/2} 
\end{split} \end{equation}
for $\eps, \delta > 0$ small enough.
\end{proposition} 

With Prop. \ref{prop:kin} with Prop. \ref{prop:pot} we can now conclude the proof of Prop. \ref{prop:Psi-energy}. 
\begin{proof}[Proof of Prop. \ref{prop:Psi-energy}] Combining (\ref{eq:Kpsi}) and (\ref{eq:Vpsi}) and introducing the notation \begin{equation}
\cE_\ell(x) :=   |\nabla f_\ell(x)|^2 +  \frac 1 2 V(x)f^2_\ell(x)
\end{equation} 
we arrive at 
\begin{equation}\begin{split}
\frac{1}{L^3} \langle &\Psi, (\cK + \cV) \Psi \rangle \\ \leq\; & \rho_0^2 \int \cE_\ell(x)\, dx  
+ 2\rho_0 \|\s\|^2  \int  \cE_\ell(x) dx  
+ \int |\nabla f_\ell (x)|^2 |\sigma (x)|^2 dx \\ &+ 2\int f_\ell (x) \nabla f_\ell (x) \sigma (x) \nabla \sigma (x) dx  +  \int |f_\ell(x)|^2 |\nabla \s(x)|^2 dx \\
& + 2\rho_0 \int \cE_\ell(x) \big[(\g-\mathbbm{1}) \ast \s)(x) + (\s \ast \s)(x)\big] dx + 2\rho_0 \int |\nabla f_\ell(x)|^2 \sigma(x) dx \\
& + 2\rho_0 \int  f_\ell(x) \nabla f_\ell(x) \cdot \nabla \s (x)  dx + C \rho^{5/2+ \delta} ,
\end{split}
	\end{equation}
if $0 < \delta \leq \eps /2$ and $\eps > 0$ is small enough. Defining 
\begin{equation}\label{def:nu}
\begin{split} 
\th (x) = &\; - \rho_0 \omega_\ell (x) + f_\ell (x) \tilde{\sigma} (x) = - \rho_0 \chi_\ell (2x)  \omega_\ell (x) + \chi_{\ell_0} (2x) (1 - \chi_\ell (2x)) s (x) 
\end{split} 
\end{equation} 
with $\tilde{\sigma}$ as in (\ref{eq:tlsigma}), we obtain 
\begin{equation}\label{eq:Hpsi} \begin{split} 
\frac{1}{L^3} \langle &\Psi, (\cK + \cV) \Psi \rangle \\ \leq\; &\| \nabla \th\|_2^2 +\frac{\rho_0^2}{2} \int V(x)f^2_\ell(x) dx  + 16 \pi \aa \rho_0 \| s\|_2^2 + 8 \pi \aa \rho_0   ((g-\mathbbm{1}) \ast \sigma)(0) + \sum_{i=1}^4 \mathcal{R}_i 
\end{split}
\end{equation} 
with 
\[\begin{split}
\mathcal{R}_1 =\; &  2\rho_0 \|\s\|_2^2  \int \cE_\ell(x) dx  - 8 \pi \aa \rho_0 \| s\|_2^2 \\
\mathcal{R}_2 = \; & 2\rho_0 \int  \cE_\ell(x) ((\gamma-\mathbbm{1}) \ast \sigma)(x) dx - 8 \pi \aa \rho_0   ((g-\mathbbm{1}) \ast \sigma)(0) \\
\mathcal{R}_3 = \; & 2\rho_0 \int \cE_\ell(x)^2 (\sigma \ast \sigma)(x) dx - 8 \pi \aa \rho_0 \| s\|_2^2
\end{split}\]
and
\[
\mathcal{R}_4 =  \| \nabla \{- \rho_0 \o_\ell +f_\ell \s\}\|_2^2 - \| \nabla \{- \rho_0 \o_\ell +f_\ell \tilde\s\}\|_2^2\,.
\]
The bound $|\cR_4| \le C \rho^{5/2+\delta}$ was established in the proof of \cite[Lemma 3.1]{BBCOS}. The estimates $|\cR_i| \le C \rho^{5/2+\delta}, i=1,\ldots,3,$ follow from \eqref{eq:en-fl}, comparing $((\gamma-1) * \sigma) (x)$ with $((\gamma-1) * \sigma) (0)$ and then with $((g-1) * \sigma) (0)$, comparing $(\sigma * \sigma) (x)$ with $(\sigma * \sigma) (0)$ and  estimating 
\begin{equation}\label{eq:sigma-s}
\big| \| \sigma \|^2 - \| s \|^2 \big| \lesssim \rho^{3/2+\eps} \,.
\end{equation} 
The details can be found below \cite[Eq.~(3.6)]{BBCOS}, with  \eqref{eq:en-fl} replacing \cite[Eq.~(2.6)]{BBCOS}. 

Next, we observe that, with $\omega = 1 - f$ and $f$ the untruncated minimizer of (\ref{eq:Ef}), 
\begin{equation}\label{eq:nabla-nu}
\| \nabla \th\|^2 +\frac{\rho_0^2}{2} \int V(x)f^2_\ell(x) dx \leq  \|\nabla (\th + {\rho_{0}}\omega \chi_{\ell_0} (2\cdot)) \|^2 + 4\pi \frak{a} \rho^2_0 + C \rho^{5/2+\eps} \,.
\end{equation}
This follows from (see \cite[Proof of Lemma 3.2]{BBCOS} for more details) \[ |\langle \nabla (\th + \rho_0 \omega \chi_{\ell_0} (2 \cdot) ) , \nabla (\rho_0 \omega \chi_{\ell_0} (2 \cdot)) \rangle| \lesssim \rho^{5/2+\eps} \]
and from the estimate
\[ \begin{split}  
\rho_0^2 & \| \nabla (\omega \chi_{\ell_0} (2 \cdot)) \|^2  +\frac{\rho_0^2}{2} \int V(x)f^2_\ell(x) dx \\
&\leq \rho_0^2 \int \cE_{\ell_0}(x) dx + C \rho^{5/2+\eps} \leq 4\pi \frak{a} \rho_0^2 \big( 1 + C \frak{a} / \ell_0) + C \rho^{5/2+\eps} \leq 4\pi \frak{a} \rho_0^2 + C \rho^{5/2+\eps}\end{split}  \]
where we used (\ref{eq:en-fl}) and the fact that $f_\ell = f$, on $\text{supp } V$. From (\ref{eq:Hpsi}) we arrive therefore at  
\[ \begin{split} \frac{1}{L^3} \langle \Psi, (\cK + \cV) \Psi \rangle \leq \; &4\pi \frak{a} \rho_0^2 +  
\|\nabla (\th + {\rho_{0}}\omega \chi_{\ell_0} (2\cdot)) \|^2 \\ &+ 16 \pi \aa \rho_0 \| s\|_2^2 + 8 \pi \aa \rho_0   ((g-\mathbbm{1}) \ast \sigma)(0) + C \rho^{5/2+\eps} \,.
\end{split} \]
From this point, we proceed as below \cite[Eq. (3.11)]{BBCOS} to conclude that 
\[\frac{1}{L^3} \langle \Psi, (\cK + \cV) \Psi \rangle \leq 4 \pi \frak{a} \rho^2 \left( 1 + \frac{128}{15 \sqrt{\pi}} (\rho \frak{a}^3)^{1/2} \right) + \sum_{i=1}^5 \mathcal{E}_i + C \rho^{5/2+\eps}\,, \]
with the error terms 
\begin{align*}
\mathcal{E}_1 =\; & \left\|\nabla (\th + {\rho_{0}}\omega \chi_{\ell_0} (2\cdot)) \right\|^2 - \left\|\nabla (s + {\rho_{0}}\omega_{\ell_0}) \right\|^2\\
  \mathcal{E}_2 =\; & 8\pi  \mathfrak{a} {\rho_0} \|s\|^2 -\frac{1}{|\L|}\sum_{k\in \L^*_+} \rho_0\widehat{V}_\mathrm{eff}(k)  \widehat s_k^2 \\
 \mathcal{E}_3 =\; & 8\pi  \mathfrak{a} \rho_0\,   ((g-\mathbbm{1}) \ast s)(0) - \frac{1}{|\L|}\sum_{k\in \L^*_+} \rho_0\widehat{V}_\mathrm{eff}(k)\,  \widehat{(g-\mathbbm{1})}(k) \widehat s_k \\
\mathcal{E}_4 = \; &  \frac{\rho_{0}}{|\L|}\sum_{k\in \L^*_+} \big( 2|k|^2\widehat \omega_{\ell_0}(k) -\widehat{V}_\mathrm{eff}(k) \big) \widehat s_k \\
\mathcal{E}_5 =\; &  \frac{\rho_0^2}{|\L|}\sum_{k\in \L^*_+}\Big(  |k|^2  \widehat \omega^2_{\ell_0}(k) -  \frac{ \widehat{V}^2_\mathrm{eff}(k)}{4 |k|^2} \Big)\,.
\end{align*} To estimate $\mathcal{E}_1$ we first rewrite, as in \cite{BBCOS}, 
\begin{equation}\label{eq:E1-Z}
\th(x) + \rho_0 \omega (x) \chi_{\ell_0} (2x) = (1 -\chi_\ell (2x)) Z(x) - (1- \chi_{\ell_0} (2x)) Z(x)
\end{equation} 
with $Z(x)= \rho_0 \omega_{\ell_0} (x) +s(x)$.  Hence 
\begin{equation}  \label{Eq:Explicit_E_2} 
   \begin{split} 
       | \mathcal{E}_1| 
         \leq \; &\left\langle \nabla Z, \left[(1  -  \chi_\ell (2\cdot) )^2 -1 \right] \nabla Z \right\rangle  \\ &+ 2 \big|  \langle (1-\chi_\ell (2\cdot)) \nabla Z, -\frac{2}{\ell} \nabla \chi (2\cdot / \ell) Z(x) - \nabla \{ (1 - \chi_{\ell_0} (2\cdot)) Z \} \rangle\big|  \\ &+ C\ell^{-2}  \|  \nabla \chi (2\cdot / \ell) Z \|^2 + C \| \nabla \{ (1 - \chi_{\ell_0} (2\cdot)) Z \} \|^2 \,. \end{split}  \end{equation} 
To bound the r.h.s. of \eqref{Eq:Explicit_E_2} we decompose $Z = Z_1 + Z_2$, where $Z_1$ is defined through the Fourier coefficients
\be \label{def:Z1}
\widehat Z_1(k) =  \rho_0 \widehat{\o}_{\ell_0}(k)- \frac{\rho_0 \widehat{V}_{\rm{eff}}(k) }{2|k|^2}\,.   \ee
From Lemma \ref{lm:eta} and expanding (\ref{eq:hatsk}) for $|k| \gg \rho^{1/2}$, we find 
\[ |\widehat{Z}_2 (k)| \lesssim \frac{\rho}{k^2} \min \big\{ 1 , \frac{\rho}{k^2} \big\} \, . \]
This easily implies 
\[ \| Z_2 \|_\infty \lesssim \| \widehat{Z}_2 \|_1 \leq C \rho^{3/2} \]
and 
\[ \| \nabla Z_2 \|_p \lesssim \| k \widehat{Z}_2 \|_{p'} \lesssim \rho^{2-\frac{3}{2p}} \]
for all $3/2 < p < \infty$. On the other hand, recalling the definition (\ref{eq:Veff}) of $V_\text{eff}$, we find 
\[ Z_1 = \rho_0 \o_{\ell_0} - \rho_0 \omega = - \rho_0 \omega (1-\chi_{\ell_0}) \, . \]
Since $1- \chi_{\ell_0}$ is supported on $|x| \geq 2\ell_0$, where $\omega (x) = \frak{a}/|x|$, we obtain
\[ \| Z_1 \|_\infty \lesssim \rho / \ell_0 \lesssim \rho^{3/2} , \qquad \| \nabla Z_1 \|_p \lesssim \rho \ell_0^{-2+3/p} \]
for all $p > 3/2$. We conclude that \begin{equation}\label{eq:bds-Z} \| Z \|_\infty \lesssim \rho^{3/2}, \qquad \| \nabla Z \|_p \lesssim \rho^{2-3/(2p)} \end{equation}
for all $3/2 < p < \infty$ (in contrast with \cite{BBCOS}, we do not have a bound for the $L^\infty$-norm of $\nabla Z$, due to the weaker decay of $\widehat{Z}_2$ at infinity). With (\ref{eq:bds-Z}), we can estimate the r.h.s. of  (\ref{Eq:Explicit_E_2}). In particular, the first term is controlled by 
\[ \langle \nabla Z , [(1  -  \chi_\ell (2\cdot) )^2 -1 ] \nabla Z \rangle \lesssim \int_{|x| \leq 4\ell} |\nabla Z (x)|^2 dx \lesssim \ell^{3/p'} \| \nabla Z \|_{2p}^2  \lesssim  \ell^{3(1-1/p)} \rho^{4-3/(2p)} \]
for any $1 < p < \infty$. Choosing $p$ large enough, we find 
\[ \langle \nabla Z , [(1  -  \chi_\ell (2\cdot) )^2 -1 ] \nabla Z \rangle \lesssim  \rho^{5/2+\eps} \]
if $\eps,\delta > 0$ are small enough. The first term on the last line of (\ref{Eq:Explicit_E_2}) can be estimated by 
\[ \ell^{-2} \| \nabla \chi (2 \cdot / \ell) Z \|^2 \lesssim \ell^{-2} \int_{|x| \leq 4\ell} |Z (x) |^2 dx \lesssim \ell \| Z \|_\infty^2 \lesssim \ell \rho^3 \lesssim \rho^{5/2+\eps} \]
if $\eps > 0$ is small enough. As for the second term on the last line, we have 
\[ \begin{split} \| &\nabla (1- \chi_{\ell_0} (2 \cdot))Z \|^2 \\ &\lesssim \ell_0^{-2} \int_{\ell_0 \leq |x| \leq 4\ell_0} |Z (x)|^2 dx + \int_{|x| \geq \ell_0} |\nabla Z (x)|^2 \\ &\lesssim \rho_0^2 \ell_0^{-2} \int_{\ell_0 \leq |x| \leq 4\ell_0} \frac{\frak{a}^2}{|x|^2} dx + \ell_0^{-2} \int_{|x| \geq \ell_0} |s(x)|^2 dx + \rho_0^2 \int_{|x| \geq \ell_0} \frac{\frak{a}^2}{|x|^4} dx + \int_{|x| \geq \ell_0} |\nabla s (x)|^2\\ &\lesssim \rho^2 \ell_0^{-1} + \ell_0^{-4} \rho^{1/2} \lesssim \rho^{5/2+\eps} \end{split} \]
if $\eps, \delta > 0$ are small enough. In the last step, we used Lemma \ref{lm:eta} to bound the norms of $s, \nabla s$ restricted on $|x| \geq \ell_0$. The terms on the second line of (\ref{Eq:Explicit_E_2}) can be estimated similarly. We conclude that $|\cE_1| \lesssim \rho^{5/2+\eps}$, if $\eps ,\delta > 0$ are small enough. 

The estimates $|\mathcal{E}_i| \leq C \rho^{5/2+\d}$, for $i=2,\ldots,5$,  can be shown as below \cite[Eq. (3.24)]{BBCOS}.

\end{proof}

\section{Potential energy of the trial state}  
\label{sec:pot}

The goal of this section is to prove Proposition \ref{prop:pot}. To this end, we apply the identity (\ref{eq:id}) to write 
\begin{equation*}
\begin{split} 
a_y a_x \Psi =\; &\sqrt{\rho_0} a_y J(x) \Psi + a_y J(x) \int dz \, \nu (x-z) a_z^* J(z) \Psi \\ = \; &\sqrt{\rho_0} f_\ell(x-y) J(x) a_y \Psi + f_\ell(x-y) J(x) \int dz \, \nu (x-z) a_y a_z^* J(z) \Psi \\ = \; &\sqrt{\rho_0} f_\ell(x-y) J(x) a_y \Psi + f_\ell(x-y)  \nu (x-y)  J(x) J(y) \Psi \\ &+ f_\ell(x-y) J(x) \int dz \, \nu (x-z) f_\ell(z-y) a_z^* J(z) a_y \Psi \,.
\end{split} \end{equation*} 
Again with (\ref{eq:id}), we obtain 
\begin{equation}
\label{eq:axayPsi}
\begin{split}  
a_y &a_x \Psi \\ = \; & \rho_0 f_\ell(x-y) J(x) J(y) \Psi + f_\ell(x-y) \nu (x-y) J(x) J(y) \Psi \\ &+ \sqrt{\rho_0} f_\ell(x-y)  \int dz \, \nu (y-z) f_\ell(x-z) f_\ell(y-z) a_z^* J(x) J(y)  J(z) \Psi \\ &+ \sqrt{\rho_0} f_\ell(x-y) \int dz  \, \nu (x-z) f_\ell(y-z) f_\ell(x-z) a_z^* J(x) J(y) J(z)  \Psi \\ &+ f_\ell(x-y)  \int dz dw \, \nu (x-z) \nu (y-w) \\ & \hspace{.3cm} \times  f_\ell(x-z) f_\ell(y-z) f_\ell(x-w) f_\ell(y-w) f_\ell(z-w) a_z^* a_w^* J(x) J(y) J(z) J(w) \Psi 
\\ =: & \sum_{j=1}^5 \phi_j (x;y) \end{split} \end{equation} 
where we observe that $\phi_4 (x;y) = \phi_3 (y;x)$. 

In the following lemmas, we estimate the different contributions of the terms $\{ \phi_j \}_{j=1}^5$ to the expectation of the potential energy in the trial state $\Psi$. 
\begin{lemma} \label{lm:phi1} 
We have
\[ \frac{1}{2L^3} \int dx dy \, V(x-y) \| \phi_1 (x;y) \|^2 \leq \frac{\rho_0^2}{2} \int dx \, V(x) f_\ell^2 (x)\,.  \] 
\end{lemma} 
\begin{proof} 
Follows immediately from $0 \leq J(x), J(y) \leq 1$. 
\end{proof} 

\begin{lemma} \label{lm:phi2}
We have 
\[ \frac{1}{2L^3} \int dx dy \, V(x-y)  \| \phi_2 (x;y) \|^2 \leq C \rho^{3}. \] 
\end{lemma} 
\begin{proof} 
From $0 \leq J(x), J(y) \leq 1$, we immediately find 
\[  \frac{1}{2L^3} \int dx dy \, V(x-y)  \| \phi_2 (x;y) \|^2 \leq  \frac{1}{2} \int dx \, V(x) f_\ell^2 (x) |\nu (x)|^2  \,.\]
Next, we recall $\nu (x) = (\gamma^{-1} * \sigma) (x) = ((\gamma^{-1} -1) * \sigma) (x) + \sigma (x)$, where (from Lemma \ref{lm:eta}) 
\[ \| (\gamma^{-1} -1) * \sigma \|_\infty \leq \| \gamma^{-1} - 1 \|_2 \| \sigma \|_2 \lesssim \rho^{3/2} \]
and, on the support of $V$, 
\[ |\sigma (x)| = \ell_0^{-3} |\ph (x/ \ell_0)| \Big| \int \tilde{\sigma} (z) dz \Big| \lesssim \rho^{3/2+3\eps/2} \,.\]
We conclude that
\[ \frac{1}{2L^3} \int dx dy \, V(x-y)  \| \phi_2 (x;y) \|^2 \lesssim \rho^3\int V(x) f_\ell^2(x)dx \lesssim \rho^3. \]
\end{proof} 

\begin{lemma} \label{lm:phi3}
For $j=3,4$, we find 
\[ \frac{1}{2L^3} \int dx dy \, V(x-y) \| \phi_j (x;y) \|^2 \leq \frac{\rho_0}{2} \| \sigma \|_2^2 \int dx \, V(x) f_\ell^2 (x) + C \rho^{5/2+\eps} \]
if $\eps , \delta > 0$ are small enough. 
\end{lemma} 
\begin{proof} 
From the definition of $\phi_3$, we find  
\[ \begin{split}  \frac{1}{2} &\int dx dy \, V(x-y) \| \phi_3 (x;y) \|^2 \\ = \; & \frac{\rho_0}{2} \int dx dy dz dz' \, V(x-y) f_\ell^2 (x-y) \nu (y-z) \nu (y-z') \\ &\hspace{.3cm} \times f_\ell(x-z) f_\ell(y-z) f_\ell(x-z') f_\ell(y-z') \langle J(x) J(y) J(z) \Psi , a_z a_{z'}^* J(x) J(y) J(z') \Psi \rangle \,. \end{split}  \]
 Hence
 \[  \begin{split}  \frac{1}{2L^3} &\int dx dy \, V(x-y) \| \phi_3 (x;y) \|^2 \\ = \; &\frac{\rho_0}{2L^3} \int dx dy dz V(x-y) f_\ell^2 (x-y) \nu^2 (y-z) f_\ell^2 (x-z) f_\ell^2 (y-z) \| J(x) J(y) J(z) \Psi \|^2 \\ &+ \frac{\rho_0}{2L^3} \int dx dy dz dz' \, V(x-y) f_\ell^2 (x-y) \nu (y-z) \nu (y-z')  f_\ell^2 (z-z')  \\ &\hspace{.3cm} \times f_\ell^2(x-z) f_\ell^2(y-z) f_\ell^2 (x-z') f_\ell^2(y-z') \langle a_{z'} \Psi ,  J^2 (x) J^2 (y) J(z) J(z') a_z \Psi \rangle \\ =: \, & \text{A}_1 + \text{A}_2  \,.
 \end{split} \]
Estimating $0 \leq J (x) \leq 1$ and $0 \leq f_\ell(x) \leq 1$ for all $x \in \bR^3$, we obtain 
 \begin{equation}\label{eq:A1-fin} \text{A}_1 \leq \frac{\rho_0}{2} \| \nu \|^2 \int V(x) f_\ell^2 (x) dx \,. \end{equation} 
To bound $\text{A}_2$, we set $R = \rho^{-1/2-3\eps}$ and we observe that the contribution from $|y-z| > R$ is negligible, thanks to the fast decay of $\nu (y-z)$, from Lemma \ref{lm:eta} (in particular, Eq.~(\ref{eq:decay}); since $m \in \bN$ is arbitrary, we can gain any power of $\rho$). More details about this argument can be found in \cite{BBCOS}, following Eq. (5.11). To bound the contribution from $|y-z| \leq R$, we notice instead that
\[ \begin{split} \rho_0 \int_{|y-z| \leq R} &dx dy dz dz' \, V(x-y) f_\ell^2(x-y)| \nu (y-z')|^2  \big\langle a_z \Psi, \big( 1 - J^2 (x) J^2 (y) J(z) J(z') \big) a_z \Psi \big\rangle  \\ &\leq \rho_0 \int_{|y-z| \leq R} dx dy dz dz' dw \, V(x-y) f_\ell^2(x-y) |\nu (y-z')|^2 \\ & \hspace{2cm} \times ( u_\ell (x-w) + u_\ell (y-w) + \omega_\ell (z-w) + \omega_\ell (z' - w))  \| a_z a_w \Psi \|^2  \end{split} \]
where we used the notation $u_\ell = 1-f_\ell^2$ and $\omega_\ell = 1- f_\ell$. Since $u_\ell (s) = 0$ for $|s| > 4\ell$ and since $V$ has compact support, we can estimate the term proportional to $u_\ell (x-w)$ by 
\[ \rho_0  \int_{|z-w| \leq C R} dx dy dz dz' dw \, V(x-y) f_\ell^2(x-y)| \nu (y-z')|^2 u_\ell (x-w) \| a_z a_w \Psi \|^2 \lesssim \rho^{3-15\eps-2\delta} L^3 \]
where we used $\| u_\ell \|_1 \lesssim \ell^2$ and $\|Vf_\ell^2\|_1\lesssim 1$, (\ref{eq:zeta-combi1}) and Lemma \ref{lm:a-bds}. The contributions proportional to $u_\ell (y-w)$, $\omega_\ell (z'-w)$ can be bounded similarly (using also the fast decay of $\nu (y-z')$ to restrict the integral to $|y-z'| \leq R$, producing an additional negligible error term). As for the contribution proportional to $u_\ell (z-w)$, we use Lemma \ref{lm:aaaaC}. We find 
\[ \rho_0 \int_{|y-z| \leq R} dx dy dz dz' dw \,V(x-y)f_\ell^2(x-y) |\nu (y-z')|^2 u_\ell (z-w) \| a_z a_w \Psi \|^2 \lesssim \rho^{3-45\eps-\delta} L^3. \]
By Cauchy-Schwarz, we conclude that 
\[ \begin{split} \Big| \rho_0 &\int dx dy dz dz' \, V(x-y) f_\ell^2 (x-y) \nu (y-z) \nu (y-z')  f_\ell^2 (z-z')  f_\ell^2(x-z) f_\ell^2(y-z)  \\ &\hspace{2cm} \times f_\ell^2 (x-z') f_\ell^2(y-z') \langle a_{z'} \Psi ,  \big( 1- J^2 (x) J^2 (y) J(z) J(z')  \big) a_z \Psi \rangle \Big| \lesssim \rho^{5/2+\eps} L^3\end{split} \]
if $\eps, \delta > 0$ are small enough. Next, we observe that
 \begin{equation}\label{eq:phi31-f} 
 \begin{split} \rho_0 &\int dx dy dz dz' \, V(x-y) f_\ell^2(x-y)|\nu (y-z)| |\nu (y-z')| \\ &\hspace{1cm} \times \big( 1 - f_\ell^2 (z-z')  f_\ell^2(x-z) f_\ell^2(y-z) f_\ell^2 (x-z') f_\ell^2(y-z') \big) \| a_z \Psi \|^2 \\ &\lesssim 
 \rho_0 \int dx dy dz dz' \, V(x-y)f_\ell^2(x-y) |\nu (y-z)| |\nu (y-z')|  \\ &\hspace{1cm} \times \big( u_\ell (z-z') + u_\ell (x-z) + u_\ell (y-z) + u_\ell (x-z') + u_\ell (y-z') \big) \| a_z \Psi \|^2 .\end{split} \end{equation}
 The term proportional to $u_\ell (z-z')$ can be bounded by 
 \[\begin{split}
     \rho_0 \int dx &dy dz dz' \, V(x-y)f_\ell^2(x-y) |\nu (y-z)| |\nu (y-z')|  u_\ell (z-z') \| a_z \Psi \|^2\\
     \lesssim\;& \rho_0\|\nu\|_\infty \int dx dy dz dz' \, V(x-y)f_\ell^2(x-y) |\nu (y-z')|  u_\ell (z-z') \| a_z \Psi \|^2 
     \lesssim \rho^{3-6\eps- \delta} L^3
 \end{split}  \]
 using Lemma \ref{lm:a-bds}, $\|Vf_\ell^2\|_1\lesssim 1$,  and $\| \nu \|_\infty \lesssim \rho / \ell$, $\| \nu \|_1 \lesssim 1$ from (\ref{eq:zeta-combi1}). The other contributions can be bounded similarly. 
We conclude that 
 \[ \begin{split} \Big| &\rho_0 \int dx dy dz dz' \, V(x-y) f_\ell^2 (x-y) \nu (y-z) \nu (y-z') \\ & \times  \big(  f_\ell^2 (z-z')  f_\ell^2(x-z) f_\ell^2(y-z) f_\ell^2 (x-z') f_\ell^2(y-z') - 1 \big) \langle a_{z'} \Psi , a_z \Psi \rangle \Big| \lesssim \rho^{5/2+\eps} L^3\end{split} \]
if $\eps,\delta > 0$ are small enough. Hence, we arrive at 
\[ \text{A}_2 \leq C \rho^{3-\delta}  + \frac{\rho_0}{2L^3} \int dx dy dz dz' V(x-y) f_\ell^2 (x-y) \nu (y-z) \nu (y-z') \langle \Psi, a^*_{z'} a_z \Psi \rangle. \]
To estimate the last term, we write $a_z = \sqrt{\rho_0} + b_z, a_{z'} = \sqrt{\rho_0} +  b_{z'}$ and we observe that, since $\int \nu = 0$, the contribution of $\sqrt{\rho_0}$ vanishes. Therefore 
\[ \text{A}_2 \leq C \rho^{3-\delta}  + \frac{\rho_0}{2L^3} \int dx dy dz dz'\, V(x-y) f_\ell^2 (x-y) \nu (y-z) \nu (y-z') \langle \Psi, b^*_{z'} b_z \Psi \rangle .\]
Writing 
\[ b^*_{z'} = c^* (\gamma_{z'}) - c (\sigma_{z'}) , \quad b_z = c (\gamma_z) - c^* (\sigma_z) \, . \]
 we obtain 
 \[ \begin{split} \text{A}_2 \leq \; &C \rho^{3-\delta}  + \frac{\rho_0}{2L^3} \int dx dy dz dz' \,V(x-y) f_\ell^2 (x-y) \nu (y-z) \nu (y-z') \langle \sigma_{z'} , \sigma_z \rangle \\ &+ 
 \frac{\rho_0}{2L^3} \int dx dy dz dz' \,V(x-y) f_\ell^2 (x-y) \nu (y-z) \nu (y-z') \\ & \hspace{1cm} \times \langle \Psi, \big( c^* (\gamma_{z'}) c (\gamma_z) + c^* (\sigma_{z'}) c (\sigma_z) - c^* (\gamma_{z'}) c^* (\sigma_z) - c (\sigma_{z'}) c (\gamma_z) \big) \Psi \rangle .\end{split} \]
 With Cauchy-Schwarz and by Lemma \ref{lm:c's}, we find 
 \[ \rho_0 \int dx dy dz dz' \, V(x-y) f_\ell^2 (x-y) |\nu (y-z)| |\nu (y-z')| \| c (\tau_z) \Psi \|^2 \lesssim \rho^{4-25\eps-\delta} L^3 \]
 for both $\tau = \gamma, \sigma$ and 
 \[ \rho_0 \int dx dy dz dz' \, V(x-y) f_\ell^2 (x-y) |\nu (y-z)|^2 \| c (\gamma_{z'}) c (\sigma_z) \Psi \|^2 \lesssim \rho^{9/2-42\eps-\delta} L^3 .\]
Thus,
\[ \text{A}_2 \leq   \frac{\rho_0}{2} (\nu* \nu * \sigma * \sigma) (0) \int V(x) f_\ell^2 (x) dx + C \rho^{5/2+\eps}  \]
if $\eps , \delta> 0$ is small enough. To combine this term with (\ref{eq:A1-fin}), we observe that $\| \nu \|^2_2 = (\nu * \nu ) (0)$ and we recall that $\nu = \gamma^{-1} * \sigma$ to rewrite 
\[ (\nu * \nu) (0) + (\nu * \nu * \sigma * \sigma) (0) = (\nu * \nu * \gamma * \gamma) (0) = (\sigma * \sigma) (0) = \| \sigma \|_2^2 .\]
This concludes the proof of the lemma. 
 \end{proof} 

\begin{lemma} \label{lm:phi5} 
We have 
\[ \frac{1}{2L^3} \int dx dy \, V(x-y) \| \phi_5 (x;y) \|^2 \leq C \rho^{5/2+\eps} \]
if $\eps, \delta > 0$ are small enough.
\end{lemma} 
\begin{proof} 
From the definition of $\phi_5$, we obtain 
\[ \begin{split}  \int dx &dy \, V(x-y) \| \phi_5 (x;y) \|^2 \\ = \; & \int dx dy dz dw dz' dw' \, V(x-y) f_\ell^2 (x-y) \nu (x-z) \nu (y-w) \nu (x-z') \nu (y-w') \\ &\hspace{.3cm} \times f_\ell(x-z) f_\ell(x-z') f_\ell(y-z)  f_\ell(y-z') f_\ell(x-w)  f_\ell(x-w') f_\ell(y-w) f_\ell(y-w')  \\ &\hspace{.3cm} \times   f_\ell(z-w)f_\ell(z'-w') \langle J(x) J(y) J(z) J(w) \Psi, a_w a_z  a^*_{z'} a_{w'}^*  J(x) J(y) J(z') J(w') \Psi \rangle . \end{split} \] 
With the canonical commutation relations, we find 
\[ \begin{split} a_w a_z a_{z'}^* a^*_{w'} = \; &\delta (z-z') \delta (w-w') + \delta (w-z') \delta (z-w') + \delta (z-z') a_{w'}^* a_w + \delta (z'-w) a_{w'}^* a_z \\ &+ \delta (z-w') a_{z'}^* a_w + \delta (w-w') a_{z'}^* a_z + a_{z'}^* a_{w'}^* a_w a_z \,. \end{split} \]
Hence, we have 
\[ \frac{1}{2L^3} \int dx dy \, V(x-y) \| \phi_5 (x;y) \|^2 = \sum_{j=1}^5 \text{B}_j  \]
with 
\[ \begin{split} \text{B}_1 = \; &\frac{1}{2L^3} \int dx dy dz dw \, V(x-y) f_\ell^2 (x-y) \nu^2 (x-z)  \nu^2 (y-w)\\ &\hspace{.2cm} \times    f_\ell^2 (x-z) f_\ell^2 (x-w)  f_\ell^2 (y-z) f_\ell^2 (y-w)  f_\ell^2 (w-z) \| J(x) J(y) J(z) J(w) \Psi \|^2 \end{split} \]
\[ \begin{split} \text{B}_2 = \; &\frac{1}{2L^3} \int dx dy dz dw \,  V(x-y) f_\ell^2 (x-y) \nu (x-z) \nu (x-w) \nu (y-w) \nu (y-z) \\ &\hspace{.2cm} \times f_\ell^2 (x-z) f_\ell^2 (x-w)  f_\ell^2 (y-z) f_\ell^2 (y-w)  f_\ell^2 (w-z) \| J(x) J(y) J(z) J(w) \Psi \|^2 \end{split} \]
\[ \begin{split} \text{B}_3 =  \; &\frac{1}{L^3} \int dx dy dz dw dw' \, V(x-y) f_\ell^2 (x-y) \nu^2 (x-z) \nu (y-w) \nu (y-w') \\ &\hspace{.2cm} \times 
f_\ell^2(x-z) f_\ell^2 (y-z) f_\ell^2 (x-w) f_\ell^2 (x-w') f_\ell^2 (y-w) f_\ell^2 (y-w') f_\ell^2 (z-w) \\ &\hspace{.2cm} \times f_\ell^2 (z-w') f_\ell^2 (w-w') \langle a_{w'} \Psi, J(x)^2 J(y)^2 J(z)^2 J(w) J(w') a_w \Psi \rangle \end{split} \]
\[ \begin{split} \text{B}_4 = \; &\frac{1}{L^3}  \int dx dy dz dw dw' \, V(x-y) f_\ell^2 (x-y) \nu (x-z) \nu (x-w) \nu (y-w) \nu (y-w') \\ &\hspace{.2cm} \times 
f_\ell^2 (x-z)   f_\ell^2 (y-z) f_\ell^2 (x-w) f_\ell^2 (x-w') f_\ell^2 (y-w) f_\ell^2 (y-w') f_\ell^2 (z-w)  \\ &\hspace{.2cm} \times f_\ell^2 (z-w') f_\ell^2 (w-w') \langle a_{w'} \Psi, J(x)^2 J(y)^2 J(w)^2 J(z) J(w') a_z \Psi \rangle  \end{split} \]
\[ \begin{split} \text{B}_5 = \; &\frac{1}{2L^3} \int dx dy dz dw dz' dw' \, V(x-y) f_\ell^2 (x-y) \nu (x-z) \nu (x-z') \nu (y-w) \nu (y-w')  \\ &\hspace{.2cm} \times 
f_\ell^2 (x-z) f_\ell^2 (x-z') f_\ell^2 (x-w)  f_\ell^2 (x-w')  f_\ell^2 (y-z) ^2 (y-w) f_\ell^2 (y-w') \\ &\hspace{.2cm} \times f_\ell^2 (y-z') f_\ell^2 (z-w') f_\ell^2 (z-z') f_\ell^2 (z'-w) f_\ell^2 (w-w') f_\ell(z-w) f_\ell(z'-w')  \\ &\hspace{.2cm} \times \langle a_{z'} a_{w'} \Psi , J (x)^2 J(y)^2 J(z) J(w) J(z') J(w') a_z a_w \Psi \rangle .
 \end{split} \]
Since $0 \leq f_\ell(s) \leq 1, 0 \leq J(s) \leq 1$ for all $s \in \bR^3$, we immediately obtain 
\[ \text{B}_1 \leq \frac{\| \nu \|_2^4}{2} \int V(x) f_\ell^2 (x) dx \lesssim \rho^3 \]
and 
\[ \text{B}_2 \leq \frac{1}{2} \int V(x) f_\ell^2 (x) (\nu * \nu)^2 (x) dx \lesssim \| \nu * \nu \|_\infty^2 \int V(x) f_\ell^2(x)dx \lesssim \| \nu \|_2^4 \lesssim \rho^3.  \]

To bound the term $\text{B}_3$, we first observe that 
\[\begin{split}  \int &dx dy dz dw dw' V(x-y) f_\ell^2(x-y)\nu^2 (x-z) |\nu (y-w)| |\nu (y-w')|\\ &\hspace{2cm} \times \langle a_{w'} \Psi, \big( 1 - J^2 (x) J^2 (y) J^2 (z) J(w) J(w') \big) a_{w'} \Psi \rangle \\ \leq \; &\int dx dy dz dw dw'ds \, V(x-y) f_\ell^2(x-y)\nu^2 (x-z) |\nu (y-w)| | \nu (y-w')| \\ &\hspace{2cm} \times \big( u_\ell (x-s) + u_\ell (y-s) + u_\ell (z-s) + u_\ell (w-s) + u_\ell (w' -s) \big) \| a_{w'} a_s \Psi \|^2. \end{split} \]
The contribution proportional to $u_\ell (x-s)$ can be estimated with (\ref{eq:zeta-combi1}) and Lemma \ref{lm:a-bds}. Using also the decay (\ref{eq:decay}) of $\nu$ to restrict the integral to $|s-w'| \leq R = \rho^{-1/2-3\eps}$,  we find  
\[ \begin{split}  \int dx &dy dz dw dw' ds \, V(x-y)f_\ell^2(x-y) \nu^2 (x-z) |\nu (y-w)| | \nu (y-w')| u_\ell (x-s) \| a_{w'} a_s \Psi \|^2 \\  \leq \; & \| \nu \|_\infty \| \nu \|_2^2 \| \nu \|_1 \| u_\ell \|_1 \|Vf_\ell^2\|_1\int_{|s-w'| \leq R} dw' ds \, \| a_{w'} a_s \Psi \|^2 + C \rho^3 L^3 \lesssim C \rho^{3-18\eps-\delta} L^3.  \end{split} \]
The contributions proportional to $u_\ell (y-s), u_\ell (z-s), u_\ell (w-s)$ can be controlled analogously. To bound the term proportional to $u_\ell (w'-s)$, we apply Lemma \ref{lm:aaaaC}. We find 
\[ \begin{split}  \int dx dy dz &dw dw' ds \, V(x-y)f_\ell^2(x-y) \nu^2 (x-z) |\nu (y-w)| | \nu (y-w')| u_\ell (w'-s) \| a_{w'} a_s \Psi \|^2 \\  \leq \; & \| \nu \|^2_1 \| \nu \|_2^2 \|Vf_\ell^2\|_1\int_{|w'-s| \leq \ell} dw'ds \, \| a_{w'} a_s \Psi \|^2 \lesssim \rho^{7/2-42\eps-\delta} L^3. \end{split} \]
By Cauchy-Schwarz, we conclude that 
\[ \begin{split} \text{B}_3 \leq \; &\frac{1}{L^3} \int  dx dy dz dw dw' \, V(x-y) f_\ell^2 (x-y) \nu^2 (x-z) \nu (y-w) \nu (y-w') \\ &\hspace{.2cm} \times 
f_\ell^2(x-z) f_\ell^2 (y-z) f_\ell^2 (x-w) f_\ell^2 (x-w') f_\ell^2 (y-w) f_\ell^2 (y-w') f_\ell^2 (z-w) \\ &\hspace{.2cm} \times f_\ell^2 (z-w') f_\ell^2 (w-w') \langle a_{w'} \Psi,  a_w \Psi \rangle  + C \rho^{5/2+\eps} \end{split} \]
if $\eps, \delta > 0$ are small enough. Next, we estimate 
\[ \begin{split}  1  - \prod_{r,s} f_\ell^2 (r-s) \leq \; &u_\ell (x-z) + u_\ell (y-z) + u_\ell (x-w) + u_\ell (x-w') +u_\ell (y-w) \\ &+ u_\ell (y-w') + u_\ell (z-w) + u_\ell (z-w') + u_\ell (w-w') \end{split} \]
where the product runs over all $r, s \in \{ x,y,z,w,w'\}$, with $r \not = s$ and $(r,s) \not = (x,y)$. Considering for example the contribution proportional to $u_\ell (x-z)$, we can bound it by 
\begin{equation} \label{eq:uxz} \begin{split}  \int dx &dy dz dw dw' \, V(x-y)f_\ell^2(x-y) \nu^2 (x-z) |\nu (y-w)| |\nu (y-w')| u_\ell (x-z) \| a_{w'} \Psi \|^2 \\  &\leq \| \nu \|_\infty^2 \int dx dy dz dw dw' \, V(x-y) f_\ell^2(x-y)|\nu (y-w)| |\nu (y-w')| u_\ell (x-z) \| a_{w'} \Psi \|^2 \\  &\leq \| \nu \|_\infty^2 \| \nu \|_1^2 \| u_\ell \|_1\|Vf_\ell^2\|_1 \int dw' \, \| a_{w'} \Psi \|^2 \leq  \rho^{3-9\eps}  L^3 \end{split} \end{equation}
with (\ref{eq:zeta-combi1}) and Lemma \ref{lm:a-bds}. As for the term proportional to $u_\ell (x-w)$, we estimate it by 
\begin{equation}\label{eq:uxw} \begin{split}  \int dx dy &dz dw dw' \, V(x-y) f_\ell^2(x-y) \nu^2 (x-z) |\nu (y-w)|\\
&\qquad\qquad\qquad\times|\nu (y-w')| u_\ell (x-w) \| a_w' \Psi \|^2 \\ 
\leq\;& \| \nu \|_\infty   \int dx dy dz dw dw' \, V(x-y)f_\ell^2(x-y) \nu^2 (x-z)\\
&\qquad\qquad\qquad\times|\nu (y-w')| u_\ell (x-w) \| a_w' \Psi \|^2 \\ \leq\;& \| \nu \|_\infty \| \nu \|_2^2 \| \nu \|_1 \| u_\ell \|_1\|Vf_\ell^2\|_1 \int dw' \, \| a_{w'} \Psi \|^2 \leq  \rho^{7/2-6\eps-\delta}  L^3. \end{split} \end{equation}
The other contributions can be bounded either similarly to (\ref{eq:uxz}) or to (\ref{eq:uxw}). Hence, 
\[\begin{split}
    \text{B}_3 \leq \;&\frac{1}{L^3} \int dx dy dz dw dw' \, V(x-y) f_\ell^2 (x-y)  \nu^2 (x-z) \nu (y-w) \nu (y-w')  \langle a_{w'} \Psi,  a_w \Psi \rangle  
\\
&+ C \rho^{5/2+\eps},
\end{split}  \]
if $\eps, \delta > 0$ are small enough. We write $a_w = \sqrt{\rho_0} + b_w, a_{w'} = \sqrt{\rho_0} + b_{w'}$ and we observe that the contribution of $\sqrt{\rho_0}$ vanishes, since $\int \nu = 0$. Therefore, with Lemma \ref{lm:a-bds} and (\ref{eq:zeta-combi1}), we obtain 
\[ \begin{split} \text{B}_3 \leq \; &C\rho^{5/2+\eps} + \frac{1}{L^3} \int  dx dy dz dw dw' \, V(x-y) f_\ell^2(x-y) \nu^2 (x-z) |\nu (y-w)| \\
&\qquad\qquad\qquad\qquad\qquad\qquad\qquad\times|\nu (y-w')| \| b_{w'} \Psi \|^2 \\
 \leq \; & C \rho^{5/2+\eps} + \frac{C \| \nu \|_2^2 \| \nu \|_1^2\|Vf_\ell^2\|_1}{L^3} \int dw' \| b_{w'} \Psi \|^2 \leq C  \rho^{5/2+\eps} \end{split}  \] 
if $\eps, \delta > 0$ are small enough. 

To estimate $\text{B}_4$, we proceed similarly to what we did in order to remove from $\text{B}_3$ the factor $J^2 (x) J^2 (y) J^2 (w) J(z) J(w')$ and all factors of $f_\ell^2$ except for $f_\ell^2(x-y)$. We arrive at 
\[  \begin{split} \text{B}_4 \leq \; &C \rho^{5/2+\eps} + \frac{1}{L^3} \int dx dy dz dw dw' \, V (x-y) f_\ell^2 (x-y)  \nu (x-z) \nu (x-w) \\ &\hspace{6cm} \times \nu (y-w) \nu (y-w') \langle a_{w'} \Psi , a_z \Psi \rangle \, . \end{split} \]
As above, in the integral we can replace $a_{w'}, a_z$ by $b_{w'}, b_z$. With (\ref{eq:zeta-combi1}) and Lemma \ref{lm:a-bds}, we conclude that 
\[ \begin{split} \text{B}_4 \leq \; &C \rho^{5/2+\eps}+\\
&+\frac{1}{L^3} \int dx dy dz dw dw' \, V(x-y)f_\ell^2(x-y) |\nu (x-w)|^2 |\nu (x-z)| |\nu (y-w')| \| b_{w'} \Psi \|^2 \\ &+ \frac{1}{L^3}  \int dx dy dz dw dw' \, V(x-y)f_\ell^2(x-y) |\nu (y-w)|^2 |\nu (x-z)| |\nu (y-w')| \| b_{z} \Psi \|^2  \\ \lesssim \; &\frac{\| \nu \|_2^2 \| \nu \|_1^2 \|Vf_\ell^2\|_1}{L^3}  \int dw' \, \| b_{w'} \Psi \|^2 + C \rho^{5/2 + \eps} \leq C \rho^{5/2+\eps} \end{split} \]
if $\eps, \delta > 0$ are small enough. 

Finally, we consider the term $\text{B}_5$. As above, we first remove the $J$-operators and the $f_\ell^2$-factors $\prod_{r,s} f_\ell^2 (r-s)$. We find 
\[ \begin{split} \text{B}_5 \leq \; &C \rho^{5/2+\eps} + \frac{1}{L^3} \int dx dy dz dw dz' dw' \, V(x-y)f_\ell^2(x-y)  \nu (x-z) \nu (x-z') \nu (y-w)  \\ &\hspace{8cm} \times \nu (y-w')\langle a_{z' }a_{w'} \Psi, a_z a_w \Psi \rangle \\ = \; &C \rho^{5/2+\eps}+ \frac{1}{L^3} \int dx dy dz dw dz' dw' \, V(x-y) f_\ell^2(x-y) \nu (x-z) \nu (x-z') \nu (y-w) \\ &\hspace{8cm} \times \nu (y-w') \langle b_{z' }b_{w'} \Psi ,b_z b_w \Psi \rangle \,. \end{split} \]
Using the fast decay (\ref{eq:decay}) of $\nu$ to restrict the integrals to $|x-z| \leq R$, $|y-w| \leq R$, Cauchy-Schwarz, (\ref{eq:zeta-combi1}) and Lemma \ref{lm:a-bds}, we obtain 
\[ \begin{split} \text{B}_5 \leq \; &C \rho^{5/2+\eps} +\frac{C}{L^3} \int_{|x-z| \leq R, |y-w| \leq R} dx dy dz dw dz' dw' \, V(x-y) f_\ell^2(x-y)\\ &\hspace{7cm} \times |\nu (x-z')|^2 |\nu (y-w')|^2  \| b_z b_w \Psi \|^2\\  \leq \; &C \rho^{5/2+\eps} + C  \frac{R^3 \| \nu \|^4_2 \|Vf_\ell^2\|_1}{L^3} \int_{|z-w| \leq 2R} dz dw \, | b_z b_w \Psi \|^2 \leq C \rho^{5/2+\eps} 
  \end{split} \]
if $\eps > 0$ is small enough. This concludes the proof of the lemma. 
\end{proof} 


\begin{lemma} We have 
\[ \frac{1}{L^3} \text{Re } \int dx dy \, V(x-y) \langle \phi_1 (x;y) , \phi_2 (x;y) \rangle \leq \rho_0 \int dx \, V(x) f_\ell^2 (x) \nu (x)  + C \rho^{5/2+\eps}  \]
if $\eps , \delta > 0$ are small enough. 
\end{lemma} 
\begin{proof} 
From the definition of $\phi_1, \phi_2$, we obtain 
\[ \int dx dy \, V(x-y) \langle \phi_1 (x;y) , \phi_2 (x;y) \rangle = \rho_0 \int dx dy \, V(x-y) f_\ell^2 (x-y) \nu (x-y) \| J(x) J(y) \Psi \|^2 . \]
With (\ref{eq:zeta-combi1}), Lemma \ref{lm:a-bds} and $\| u_\ell \|_1 \lesssim \rho^{-2\delta}$, we find 
\[ \begin{split} \Big| \rho_0 \int dx dy \,& V(x-y) f_\ell^2 (x-y) \nu (x-y) \langle \Psi, \big( 1 - J^2 (x) J^2 (y) \big) \Psi \rangle \Big| \\ &\leq \rho_0 \int dx dy ds \, V(x-y) f_\ell^2(x-y)|\nu (x-y)| \big(u_\ell (x-s) + u_\ell (y-s) \big) \| a_s \Psi \|^2 \\ &\leq C \rho_0 \| \nu \|_\infty \| u_\ell \|_1 \|Vf_\ell^2\|_1 \int ds \, \| a_s \Psi \|^2  \leq C \rho^{5/2+\eps} L^3 \end{split} \]
if $\eps , \delta > 0$ are small enough. Hence,
\[ \frac{1}{L^3} \text{Re} \int dx dy \, V(x-y) \langle \phi_1 (x;y) , \phi_2 (x;y) \rangle \leq \rho_0 \int dx \, V(x) f_\ell^2 (x) \nu (x) + C \rho^{5/2+\eps} .\]
\end{proof}

\begin{lemma}\label{lm:phi1phi3} 
For $j=3,4$, we have 
\[ \frac{1}{L^3} \text{Re } \int dx dy \, V(x-y) \langle \phi_1 (x;y) , \phi_j (x;y) \rangle \leq C \rho^{5/2+\eps} \]
if $\eps,\delta > 0$ are small enough.
\end{lemma} 
\begin{proof} 
We consider only $j=3$, the case $j=4$ can be handled analogously. From the definition of $\phi_1, \phi_3$, we have
\[ \begin{split}  &\int dx dy \, V(x-y) \langle \phi_1 (x) , \phi_3 (x) \rangle \\ &= \rho_0^{3/2} \int dx dy dz \, V(x-y) f_\ell^2 (x-y) \nu (y-z) f_\ell^2 (y-z) f_\ell^2 (x-z) \langle a_z \Psi , J^2 (x) J^2 (y) J(z) \Psi \rangle .\end{split} \]
To estimate this contribution, we observe that, applying (\ref{eq:decay}) to restrict the integral to $|y-z| \leq R$ and then using Cauchy-Schwarz, 
\[ \begin{split} 
\Big|  \rho_0^{3/2} &\int dx dy dz \, V(x-y) f_\ell^2 (x-y) \nu (y-z) f_\ell^2 (y-z) f_\ell^2 (x-z) \\ &\hspace{7cm} \times  \langle a_z \Psi , \big(1 - J^2 (x) J^2 (y) J(z)\big)  \Psi \rangle \Big| \\  \leq \; &C \rho^{5/2+\eps} + \rho^{3/2} \Big( \int_{|y-z| \leq R} dx dy dz \, V(x-y)f_\ell^2(x-y) \\
&\hspace{5cm}\times\big( u_\ell (x-s) + u_\ell (y-s) + u_\ell (z-s) \big) \| a_z a_s \Psi \|^2 \Big)^{1/2} \\ &\hspace{.2cm} \times \Big( \int dx dy dz \, V(x-y)f_\ell^2(x-y) |\nu (y-z)|^2 \\
&\hspace{4cm}\times\big( u_\ell (x-s) + u_\ell (y-s) + u_\ell (z-s) \big) \| a_s \Psi \|^2 \Big)^{1/2} \\ \leq \; &C \rho^{5/2+\eps} + C \rho^{3/2} \|Vf_\ell^2\|_1 \bigg( \| u_\ell \|_1 \int_{|z-s| \leq R} dz ds \, \| a_z a_s \Psi \|^2\\
&\hspace{3.5cm}+ R^3 \int_{|z-s| \leq \ell} dz ds \, \| a_z a_s \Psi \|^2 \bigg)^{1/2} \bigg( \| \nu \|_2^2 \| u_\ell \|_1 \int ds \| a_s \Psi \|^2 \bigg)^{1/2} \\ \leq \; &C \rho^{5/2+\eps} 
\end{split} \]
if $\eps,\delta > 0$ are small enough, where, in the last step, we used (\ref{eq:zeta-combi1}), Lemma \ref{lm:a-bds} and Lemma~\ref{lm:aaaaC}. Furthermore, we notice that
\[ \begin{split} 
\Big|  \rho_0^{3/2} &\int dx dy dz \, V(x-y) f_\ell^2 (x-y) \nu (y-z) \big( 1 - f_\ell^2 (y-z) f_\ell^2 (x-z) \big) \langle a_z \Psi ,  \Psi \rangle \Big| \\ &\leq  \rho^{3/2} \int dx dy dz \, V(x-y)f_\ell^2(x-y)  |\nu (y-z)| \big( u_\ell (y-z) + u_\ell (x-z) \big) \| a_z \Psi \|  \\ &\leq \rho^{3/2} \| \nu \|_\infty \| u_\ell \|_1 \|Vf_\ell^2\|_1\int dz \| a_z \Psi \| \leq \rho^{5/2-\delta} L^{3/2} \Big( \int dz \, \|a_z \Psi \|^2 \Big)^{1/2} \leq \rho^{5/2+\eps} L^3 \end{split} \]
if $\eps, \delta > 0$ are small enough. The claim now follows from the observation that 
\[ \int dx dy dz \, V(x-y) f_\ell^2 (x-y) \nu (y-z) \langle a_z \Psi ,  \Psi \rangle = 0 \]
since $\int \nu = 0$. 
\end{proof}

\begin{lemma} \label{lm:phi1phi5} 
We have 
\[ \frac{1}{L^3} \text{Re } \int dx dy \, V(x-y) \langle \phi_1 (x;y) , \phi_5 (x;y) \rangle \leq \rho_0 \int dx \, V(x) f_\ell^2 (x) (\nu * \sigma * \sigma) (x) +C \rho^{5/2+\eps} \]
if $\eps, \delta > 0$ are small enough. 
\end{lemma} 
\begin{proof} 
From the definition of $\phi_1, \phi_5$, we have 
\[ \begin{split} 
 \int &dx dy \, V(x-y) \langle \phi_1 (x;y) , \phi_5 (x;y) \rangle \\ &= \rho_0 \int dx dy dz dw \, V(x-y) f_\ell^2 (x-y) \nu (x-z) \nu (y-w) \\ &\hspace{.5cm} \times  f_\ell^2 (x-z) f_\ell^2 (x-w) f_\ell^2 (y-z) f_\ell^2 (y-w) f_\ell(z-w) \langle a_z a_w \Psi, J^2 (x) J^2 (y) J(z) J(w) \Psi \rangle .\end{split} \]
 To bound this contribution, we first observe that 
 \[ \begin{split} 
 \Big| \rho_0 \int &dx dy dz dw \, V(x-y) f_\ell^2 (x-y) \nu (x-z) \nu (y-w)  f_\ell^2 (x-z) f_\ell^2 (x-w)  \\ &\hspace{.5cm} \times f_\ell^2 (y-z) f_\ell^2 (y-w) f_\ell(z-w) \langle a_z a_w \Psi, \big(1 - J^2 (x) J^2 (y) J(z) J(w) \big) \Psi \rangle \Big| \\
 \leq \; &C \rho^{5/2+\eps} L^3 + \rho_0 \Big( \int_{|y-w| , |x-z| \leq R} dx dy dz dw ds \, V(x-y)f_\ell^2(x-y) \, \\ &\hspace{3cm} \times \big(u_\ell (x-s) + u_\ell (y-s) + u_\ell (z-s) + u_\ell (w-s) \big) \| a_s a_w a_z \Psi \|^2 \Big)^{1/2} \\ &\times \Big(   \int_{|y-w| , |x-z| \leq R} dx dy dz dw ds \, V(x-y) f_\ell^2(x-y)\, |\nu (x-z)|^2 |\nu (y-w)|^2 \\ &\hspace{3cm} \times \big(u_\ell (x-s) + u_\ell (y-s) + u_\ell (z-s) + u_\ell (w-s) \big) \| a_s \Psi \|^2 \Big)^{1/2}  \\ \leq \; &C \rho^{5/2+\eps} L^3 + C \rho_0 \|Vf_\ell^2\|_1 \Big( \| u_\ell \|_1 \int_{|s-w|, |s-z| \leq C R}  ds dw dz \, \| a_s a_w a_z \Psi \|^2 \\ &\hspace{1cm} + R^3 \int_{|s-z| \leq \ell, |s-w| \leq CR} ds dw dz \| a_s a_w a_z \Psi \|^2 \Big)^{1/2} \Big( \| \nu \|_2^4 \| u_\ell \|_1 \int ds \, \| a_s \Psi \|^2 \Big)^{1/2} \\ \leq \; &C \rho^{5/2+\eps} L^3  
 \end{split} \]
if $\eps , \delta > 0$ are small enough. Here we used (\ref{eq:zeta-combi1}), Lemma \ref{lm:a-bds} and Lemma \ref{lm:aaaaC}. We can also estimate 
\[ \begin{split} 
 \Big| \rho_0 \int &dx dy dz dw \, V(x-y) f_\ell^2 (x-y) \nu (x-z) \nu (y-w) \\ &\hspace{.5cm}\times \big( 1-  f_\ell^2 (x-z) f_\ell^2 (x-w) f_\ell^2 (y-z) f_\ell^2 (y-w) f_\ell(z-w)\big)  \langle a_z a_w \Psi, \Psi \rangle \Big| \\ \leq \; &C \rho^{5/2+\eps} L^3+ \rho_0 \Big( \int_{|y-w| , |x-z| \leq R} dx dy dz dw \, V(x-y)f_\ell^2(x-y) \, \\ &\hspace{1cm} \times \big(u_\ell (x-z) + u_\ell (x-w) + u_\ell (y-z) + u_\ell (y-w) +u_\ell (z-w) \big) \| a_w a_z \Psi \|^2 \Big)^{1/2} \\ &\times \Big(   \int_{|y-w| , |x-z| \leq R} dx dy dz dw \, V(x-y)f_\ell^2(x-y) \, |\nu (x-z)|^2 |\nu (y-w)|^2 \\ &\hspace{1cm} \times  \big(u_\ell (x-z) + u_\ell (x-w) + u_\ell (y-z) + u_\ell (y-w) +u_\ell (z-w) \big) \Big)^{1/2}  \\ \leq \; &C \rho^{5/2+\eps}L^3 + C \rho_0 \|Vf_\ell^2\|_1\Big( \| u_\ell \|_1 \int_{|w-z| \leq C R}  dw dz \, \| a_w a_z \Psi \|^2 \\
 &\hspace{4cm}+ R^3 \int_{|w-z| \leq \ell} dw dz \| a_w a_z \Psi \|^2 \Big)^{1/2}  \Big( \| \nu \|_\infty^2 \| \nu \|_2^2 \| u_\ell \|_1 L^3 \Big)^{1/2} \\ \leq \; &C \rho^{5/2+\eps} L^3
 \end{split} \]
 if $\eps, \delta > 0$ are small enough. It follows that 
 \[ \begin{split}  \text{Re } &\int dx dy \, V(x-y) \langle \phi_1 (x;y) , \phi_5 (x;y) \rangle \\ &\leq \text{Re } \rho_0  \int dx dy dz dw \, V(x-y) f_\ell^2 (x-y) \nu (x-z) \nu (y-w) \langle a_z a_w \Psi, \Psi \rangle + C \rho^{5/2+\eps} L^3  \\ &= \text{Re } \rho_0 \int dx dy dz dw \, V(x-y) f_\ell^2 (x-y) \nu (x-z) \nu (y-w) \langle b_z b_w \Psi, \Psi \rangle + C \rho^{5/2+\eps} L^3 \end{split} \]
 because $\int \nu = 0$. Next, we write $b_s = c(\gamma_s) - c^* (\sigma_s)$, for $s = z,w$. We conclude that 
 \begin{equation}\label{eq:IVa} \begin{split}  \text{Re } &\int dx dy \, V(x-y) \langle \phi_1 (x;y) , \phi_5 (x;y) \rangle \\ &\leq  C \rho^{5/2+\eps} L^3 + \text{Re } \rho_0  \int dx dy dz dw \, V(x-y) f_\ell^2 (x-y) \nu (x-z) \nu (y-w)  \langle \gamma_z ; \sigma_w \rangle \\ &+\text{Re } \rho_0  \int dx dy dz dw \, V(x-y) f_\ell^2 (x-y) \nu (x-z) \nu (y-w) \\ &\hspace{2cm} \times \langle \big( c^* (\sigma_z) c (\gamma_w) + c^* (\sigma_w) c (\gamma_z) + c^* (\sigma_z) c^* (\sigma_w) + c(\gamma_z) c (\gamma_w)\big) \Psi, \Psi \rangle . \end{split} \end{equation} 
With Lemma \ref{lm:c's}, we remark that 
\[ \begin{split} 
\Big| &\rho_0 \int dx dy dz dw \, V(x-y) f_\ell^2 (x-y) \nu (x-z) \nu (y-w) \langle c (\gamma_w) \Psi , c (\sigma_z) \Psi \rangle \Big| \\ &\leq \rho_0 \int dx dy dz dw \, V(x-y)f_\ell^2(x-y) |\nu (x-z)| |\nu (y-w)| \left[ \| c (\gamma_w) \Psi \|^2 + \| c (\sigma_z) \Psi \|^2 \right] \\ &\leq C \rho \| \nu \|_1^2 \|Vf_\ell^2\|_1 \sum_{\tau = \gamma, \sigma} \int dw \, \| c (\tau_w) \Psi \|^2 \leq C \rho^{5/2+\eps} L^3 \end{split} \]
if $\eps , \delta > 0$ are small enough. Furthermore,
\[ \begin{split} 
&\Big| \rho_0 \int dx dy dz dw \, V(x-y) f_\ell^2 (x-y) \nu (x-z) \nu (y-w) \langle c (\gamma_w) c (\gamma_z) \Psi , \Psi \rangle \Big| \\ &\leq C \rho^{5/2+\eps} L^3 + \rho_0 \Big( \int_{|w-y| \leq R}  dx dy dz dw \, V(x-y) f_\ell^2(x-y) |\nu (x-z)|  \| c (\gamma_w) c (\gamma_z) \Psi \|^2 \Big)^{1/2}\\ &\hspace{3cm} \times  \Big(  \int dx dy dz dw \, V(x-y) f_\ell^2(x-y) |\nu (x-z)| |\nu (y-w)|^2 \Big)^{1/2}  \\ &\leq C \rho^{5/2+\eps} L^3 + C \rho_0 \|Vf_\ell^2\|_1 \Big(\| \nu \|_1 \int_{|w-z| \leq R} dw dz \, \| c (\gamma_w) c (\gamma_z) \Psi \|^2 \Big)^{1/2} (\| \nu \|_1 \| \nu \|_2^2 L^3)^{1/2} \\ &\leq C \rho^{5/2+\eps} L^3 . \end{split} \]
From \eqref{eq:IVa}, we find that 
\[ \begin{split} 
 \text{Re } \int dx dy \, V(x-y) \langle &\phi_1 (x;y) , \phi_5 (x;y) \rangle \\ &\leq  C \rho^{5/2+\eps} L^3 + \rho_0 L^3  \int dx \, V(x) f_\ell^2 (x) (\nu * \gamma * \sigma * \nu) (x)  
  \end{split} \] 
  which concludes the proof of the lemma, since $\nu * \nu * \gamma * \sigma = \nu * \sigma * \sigma$. 
\end{proof}

\begin{lemma} For $j=3,4$, we have 
\[ \frac{1}{L^3} \text{Re } \int dx dy \, V(x-y) \langle \phi_2 (x;y) , \phi_j (x;y) \rangle \leq C \rho^{5/2+\eps} \]
if $\eps ,\delta > 0$ are small enough. 
\end{lemma} 
\begin{proof}
The proof is very similar to the proof of Lemma \ref{lm:phi1phi3} because in all error terms we can estimate the factor $|\nu (x-y)| \lesssim \| \nu \|_\infty \lesssim \rho$ appearing in $\phi_2$, since $Vf_\ell^2$ can be used to integrate in the $x-y$ variable.  
\end{proof} 

\begin{lemma} We have 
\[ \frac{1}{L^3} \text{Re } \int dx dy \, V(x-y) \langle \phi_2 (x;y) , \phi_5 (x;y) \rangle \leq C \rho^{5/2+\eps}  \]
if $\eps, \delta > 0$ are small enough. 
\end{lemma} 
\begin{proof} 
From the definition of $\phi_2, \phi_5$, we find 
\[ \begin{split} 
 \int dx dy \, &V(x-y) \langle \phi_2 (x;y) , \phi_5 (x;y) \rangle \\ = \; & \int dx dy dz dw \, V(x-y) f_\ell^2 (x-y) \nu (x-y) \nu (x-z) \nu (y-w) f_\ell^2 (x-z) f_\ell^2 (y-z) 
 \\  &\hspace{2cm} \times f_\ell^2 (x-w) f_\ell^2 (y-w) f_\ell(z-w) \langle a_z a_w \Psi, J^2 (x) J^2 (y) J(z) J(w) \Psi \rangle. \end{split} \]
We proceed as in the proof of Lemma \ref{lm:phi1phi5} to get rid of the $J$-operators and of the $f_\ell$-factors. After estimating $|\nu (x-y)| \leq \| \nu \|_\infty \lesssim \rho$, all these error terms can be handled exactly as those arising in Lemma \ref{lm:phi1phi5}. Writing $a_z = \sqrt{\rho_0} + b_z$, $a_w = \sqrt{\rho_0} + b_w$ and observing that the contribution of the factors $\sqrt{\rho_0}$ vanishes, because $\int \nu = 0$, we obtain 
\[ \begin{split}  \frac{1}{L^3} &\text{Re} \int dx dy \, V(x-y) \langle \phi_2 (x;y) , \phi_5 (x;y) \rangle \\  \leq\;& C \rho^{5/2+\eps} + \frac{1}{L^3} \int dx dy dz dw \, V(x-y) f_\ell^2 (x-y) \nu (x-y) \nu (x-z) \nu (y-w) \langle b_z b_w \Psi , \Psi \rangle \end{split} \]
if $\eps , \delta > 0$ are small enough. Next we write $b_z = c (\gamma_z) - c^* (\sigma_z)$, $b_w = c(\gamma_w) - c^* (\sigma_w)$. The contribution of all normal ordered terms can bounded analogously as we did with the last term on the r.h.s. of (\ref{eq:IVa}) (again, after estimating $|\nu (x-y)| \lesssim \rho$). We conclude that 
\[  \begin{split}  \frac{1}{L^3} \text{Re } &\int dx dy \, V(x-y) \langle \phi_2 (x;y) , \phi_5 (x;y) \rangle \\ \leq\;& C \rho^{5/2+\eps} + \frac{1}{L^3} \int dx dy dz dw \, V(x-y) f_\ell^2 (x-y) \nu (x-y) \nu (x-z) \nu (y-w) \langle \gamma_z ; \sigma_w \rangle \\  \leq\;& C \rho^{5/2+\eps} + \int dx \, V(x) f_\ell^2 (x) \nu (x) (\nu * \gamma *\sigma *\nu) (x) \\ \leq\;& C \rho^{5/2+\eps} + \int dx \, V(x) f_\ell^2 (x) \nu (x) (\nu * \sigma * \sigma) (x) \end{split} \]
if $\eps , \delta > 0$ are small enough. As argued in the proof of Lemma \ref{lm:phi2}, we have 
\[ |\nu (x)| = |(\gamma^{-1} * \sigma ) (x) | \leq \| (\gamma^{-1} - 1)*\sigma \|_\infty + |\sigma (x)| \lesssim \rho^{3/2} \] 
on the support of $V$. Since $\nu * \sigma * \sigma = \gamma * \sigma - \nu$, we similarly find 
\[ | (\nu * \sigma * \sigma) (x)| \lesssim \rho^{3/2} \]
on the support of $V$. We conclude that 
\[ \frac{1}{L^3} \text{Re} \int dx dy \, V(x-y) \langle \phi_2 (x;y) , \phi_5 (x;y) \rangle \leq C \rho^{5/2+\eps} \]
if $\eps , \delta > 0$ are small enough. 
\end{proof} 

\begin{lemma} We have 
\[ \frac{1}{L^3} \text{Re } \int dx dy \, V(x-y) \langle \phi_3 (x;y) , \phi_4 (x;y) \rangle \leq \rho_0 \int dx \, V(x) f_\ell^2 (x) (\sigma* \sigma )( x) + C \rho^{5/2+\eps} \] if $\eps , \delta > 0$ are small enough. 
\end{lemma} 
\begin{proof} 
From the definition of $\phi_3, \phi_4$, we find 
\[ \begin{split}  \int \; &dx dy \, V(x-y) \langle \phi_3 (x;y) , \phi_4 (x;y) \rangle \\ =\;& \rho_0 \int dx dy dz dw \, V(x-y) f_\ell^2 (x-y) \nu (y-z) \nu (x-w) \\ &\hspace{1cm} \times  f_\ell(x-z) f_\ell(y-z) f_\ell(x-w) f_\ell(y-w) \langle J(x) J(y) J(z) \Psi , a_z a_w^* J(x) J(y) J(w) \Psi \rangle \,. \end{split} \]
With $a_z a_w^* = a_w^* a_z + \delta (z-w)$, we obtain 
\[  \int dx dy \, V(x-y) \langle \phi_3 (x;y) , \phi_4 (x;y) \rangle = \text{D}_1 + \text{D}_2 \]
with 
\[ \begin{split} \text{D}_1 =\;& \rho_0 \int dx dy dz \, V(x-y) f_\ell^2 (x-y) \nu (x-z) \nu (y-z) f_\ell^2(x-z) f_\ell^2(y-z)\\
&\hspace{2cm}\times  \langle \Psi, J(x)^2 J(y)^2 J(z)^2 \Psi \rangle  \\ 
\text{D}_2 =\;& \rho_0 \int dx dy dz dw \, V(x-y) f_\ell^2 (x-y) \nu (x-w) \nu (y-z) f_\ell^2 (x-z) f_\ell^2 (x-w) f_\ell^2 (y-z) \\
&\hspace{2.6cm} \times f_\ell^2 (y-w) f_\ell^2 (z-w) \langle a_w \Psi, J(x)^2 J(y)^2 J(z) J(w) a_z \Psi \rangle . \end{split} \]
To estimate $\text{D}_1$, we observe, with (\ref{eq:zeta-combi1}) and Lemma \ref{lm:a-bds}, that 
\[ \begin{split} 
\Big| \rho_0 &\int dx dy dz \, V(x-y) f_\ell^2 (x-y) \nu (x-z) \nu (y-z) \\
&\hspace{1cm}\times f_\ell^2(x-z) f_\ell^2(y-z)  \langle \Psi, \big( 1 - J^2 (x) J^2 (y) J^2 (z) \big) \Psi \rangle \Big| \\ \leq \; &
\rho \int dx dy dz ds \, V(x-y)f_\ell^2(x-y) |\nu (x-z)|^2 \big (u_\ell (x-s) + u_\ell (y-s) + u_\ell (z-s) \big) \| a_s \Psi \|^2 \\ \lesssim \; &\rho \| \nu \|_2^2 \| u_\ell \|_1 \|Vf_\ell^2\|_1\int ds \, \| a_s \Psi \|^2 \lesssim \rho^{7/2-3\eps-2\delta} L^3 \lesssim \rho^{5/2+\eps} L^3 \end{split} \]
if $\eps , \delta > 0$ are small enough. Moreover,
\begin{equation*}
    \begin{split}
        \bigg|\rho_0\int dxdydz\,& V(x-y) f_\ell^2(x-y) \nu(x-z)\nu(y-z)\big[ 1-f_\ell^2(x-z) f_\ell^2(y-z) \big]\bigg|\\
        \lesssim\;& \rho_0 \|\nu\|_\infty^2 \int dxdydz\,V(x-y) f_\ell^2(x-y) \big[ u_\ell (x-z) + u_\ell (y-z)\big] \lesssim \rho^{5/2+\varepsilon} L^3.
    \end{split}
\end{equation*}
Hence
\begin{equation} \label{eq:D1-bd} \text{D}_1 \leq \rho_0 L^3 \int dx \, V(x) f_\ell^2 (x) (\nu * \nu) (x) + C \rho^{5/2+\eps} L^3 \, .\end{equation} 
As for $\text{D}_2$, we notice that 
 \[  \begin{split} \Big| \rho_0 \int &dx dy dz dw \, V(x-y) f_\ell^2 (x-y) \nu (x-w) \nu (y-z) f_\ell^2 (x-z) f_\ell^2 (x-w) f_\ell^2 (y-z) \\ &\hspace{.5cm} \times f_\ell^2 (y-w) f_\ell^2 (z-w) \langle a_w \Psi, \big( 1 - J(x)^2 J(y)^2 J(z) J(w) \big) a_z \Psi \rangle \Big| \\ \leq \; &C \rho^{5/2+\eps} L^3 + C \rho \int_{|y-z| \leq R}  dx dy dz dw ds \, V(x-y)f_\ell^2(x-y) |\nu (x-w)|^2  \\ &\hspace{3cm} \times \big( u_\ell (x-s) + u_\ell (y-s) + u_\ell (z-s) + u_\ell (w-s) \big) \| a_s a_z \Psi \|^2 
 \\ \leq \; &C \rho^{5/2+\eps} L^3 
  \end{split} \]
  and that 
  \[ \begin{split} \Big| \rho_0 \int &dx dy dz dw \, V(x-y) f_\ell^2 (x-y) \nu (x-w) \nu (y-z)  \\ &\hspace{.5cm} \times \big( 1 - f_\ell^2 (x-z) f_\ell^2 (x-w) f_\ell^2 (y-z) f_\ell^2 (y-w) f_\ell^2 (z-w) \big)  \langle a_w \Psi, a_z \Psi \rangle \Big| \\ \leq \; &C \rho^{5/2+\eps} L^3 + C \rho \int  dx dy dz dw \, V(x-y) f_\ell^2(x-y)|\nu (x-w)| |\nu (y-z)|   \\ &\hspace{.5cm} \times \big( u_\ell (x-z) + u_\ell (x-w) + u_\ell (y-z) + u_\ell (y-w) + u_\ell (z-w)  \big) \| a_z \Psi \|^2 
 \\ \leq \; &C \rho^{5/2+\eps} L^3 
  \end{split} \]
if $\eps , \delta > 0$ are small enough. Hence
\[ \begin{split}  \text{D}_2 \leq \; &\rho_0 \int dx dy dz dw \, V(x-y) f_\ell^2 (x-y) \nu (x-w) \nu (y-z) \langle b_w \Psi, b_z \Psi \rangle + C \rho^{5/2+\eps} \\ 
\leq \; &\rho_0 \int dx dy dz dw \, V(x-y) f_\ell^2 (x-y) \nu (x-w) \nu (y-z) \langle \sigma_w ; \sigma_z \rangle \\ &+ 
\rho_0 \int dx dy dz dw \, V(x-y) f_\ell^2 (x-y) \nu (x-w) \nu (y-z) \\ &\hspace{1cm} \times \langle \Psi, \big( c^* (\gamma_w) c (\gamma_z) + c^* (\sigma_z) c (\sigma_w) + c^* (\gamma_w) c^* (\sigma_z) + c (\sigma_w) c (\gamma_z) \big) \Psi \rangle \\ &+C \rho^{5/2+\eps} .
\end{split} \] 
With (\ref{eq:zeta-combi1}) and Lemma \ref{lm:c's}, we find 
\[ \rho \int dx dy dz dw \, V(x-y)f_\ell^2(x-y) |\nu (x-w)| |\nu (y-z)| \big[ \| c (\gamma_z) \Psi \|^2 + \| c (\sigma_z) \Psi \|^2 \big] \leq C \rho^{5/2+\eps} L^3 \]
and 
\[ \begin{split} \rho \int dx dy dz &dw \, V(x-y) f_\ell^2(x-y) |\nu (x-w)| |\nu (y-z)|  \| c (\gamma_w) c (\sigma_z) \Psi \| \\ \leq \; &\rho \Big( \int dx dy dz dw \, V(x-y)f_\ell^2(x-y) |\nu (x-w)| \|c (\gamma_w) c (\sigma_z) \Psi \|^2 \Big)^{1/2} \\ &\hspace{2cm} \times  \Big(  \int dx dy dz dw \, V(x-y) f_\ell^2(x-y)|\nu (x-w)|  |\nu (y-z)|^2 \Big)^{1/2} \\ \leq \; &C \rho^{5/2+\eps} L^3 \end{split} \]
if $\eps, \delta > 0$ are small enough. Therefore
\[ \begin{split}  \text{D}_2 \leq \; &\rho_0 \int dx\, V(x) f_\ell^2 (x)  (\nu * \nu * \sigma * \sigma) (x) + C \rho^{5/2+\eps} L^3. \end{split} \] 
With (\ref{eq:D1-bd}) and since $1 + \sigma*\sigma = \gamma * \gamma$, we conclude that 
\[\frac{1}{L^3}  \int dx dy \, V(x-y) \langle \phi_3 (x,y) \Psi ; \phi_4 (x,y) \Psi \leq \rho_0  \int dx \, V(x) f_\ell^2 (x) (\sigma * \sigma) (x) + C \rho^{5/2+\eps}.\] 
\end{proof} 

\begin{lemma}\label{lm:phi3phi5} For $j=3,4$, we have 
\[ \frac{1}{L^3} \text{Re } \int dx dy \, V(x-y) \langle \phi_j (x;y) , \phi_5 (x;y) \rangle \leq C \rho^{5/2+\eps/2} \]
if $\eps , \delta > 0$ are small enough. 
\end{lemma} 
\begin{proof} 
From Lemma \ref{lm:phi3}, we have 
\[ \frac{1}{L^3} \int dx dy \, V(x-y) \| \phi_j (x;y) \|^2 \leq C \rho^{5/2} \]
for $j=3,4$. From Lemma \ref{lm:phi5}, we also find 
\[ \frac{1}{L^3} \int dx dy \, V(x-y) \| \phi_5 (x;y) \|^2 \leq C \rho^{5/2+\eps} \]
for $\eps, \delta > 0$ small enough. By Cauchy-Schwarz, we conclude that 
\[ \frac{1}{L^3} \Big| \int dx dy \, V(x-y) \langle \phi_j (x,y) ; \phi_5 (x,y) \rangle \Big| \leq C \rho^{5/2+\eps/2}  \]
for $\eps , \delta > 0$ small enough. 
\end{proof} 

Combining Lemmas \ref{lm:phi1} - \ref{lm:phi3phi5}, we obtain the desired upper bound on the potential energy in the trial state $\Psi$.
\begin{proof}[Proof of Prop. \ref{prop:pot}] 
Collecting all estimates from Lemmas \ref{lm:phi1} - \ref{lm:phi3phi5}, we arrive at 
\[ \begin{split} 
\frac{1}{2L^3} \int &dx dy \, V(x-y) \| a_x a_y \Psi \|^2 \\ \leq \; &\frac{\rho_0^2}{2} \int dx \, V(x) f_\ell^2 (x) + \rho_0 \| \sigma \|_2^2 \int dx \, V(x) f_\ell^2 (x) \\ &+ \rho_0 \int dx \, V(x) f_\ell^2 (x) \nu (x) + \rho_0 \int dx V(x) f_\ell^2 (x) (\nu * \sigma * \sigma) (x) \\ &+ \rho_0 \int dx V(x) f_\ell^2 (x) (\sigma * \sigma) (x) + C \rho^{5/2+ \eps/2} .
\end{split} \]
With $\nu + \nu * \sigma * \sigma = \sigma * \gamma$, we obtain 
\[ \begin{split} 
\frac{1}{2L^3} \int &dx dy \, V(x-y) \| a_x a_y \Psi \|^2 \\ \leq \; &\frac{\rho_0^2}{2} \int dx \, V(x) f_\ell^2 (x) + \rho_0 \| \sigma \|_2^2 \int dx \, V(x) f_\ell^2 (x) \\ &+ \rho_0 \int dx \, V(x) f_\ell^2 (x)  \big( (\sigma * \sigma) (x) + (\sigma * \gamma) (x) \big) + C \rho^{5/2+ \eps/2} .
\end{split} \]
In the last term, we decompose $\sigma * \gamma = \sigma * (\gamma-1) + \sigma$ and we observe that, on the support of $V$, 
\[ |\sigma (x)| \leq \ell_0^{-3} |\ph (x/\ell_0)| \Big| \int \tilde{\sigma} (z) dz \Big| \lesssim \rho^{3/2+3\eps/2}. \]
Hence 
\[ \begin{split} 
\frac{1}{2L^3} \int &dx dy \, V(x-y) \| a_x a_y \Psi \|^2 \\ \leq \; &\frac{\rho_0^2}{2} \int dx \, V(x) f_\ell^2 (x) + \rho_0 \| \sigma \|_2^2 \int dx \, V(x) f_\ell^2 (x) \\ &+ \rho_0 \int dx \, V(x) f_\ell^2 (x)  \big( (\sigma * \sigma) (x) + (\sigma * (\gamma-1)) (x) \big) + C \rho^{5/2+ \eps/2} .
\end{split} \]

%
\end{proof}

\appendix

\section{Properties of the kernels: proof of Lemma \ref{lm:eta}}
\label{app:kernels}

We proceed similarly as in the proof of \cite[Lemma 2.2]{BBCOS}. The main difference is that now we can only rely on the estimates $\| \nabla^m \widehat{V}_\text{eff} \|_\infty \leq C_m$, for every $m \in \bN$, from Lemma \ref{Lem:general_potential_eff} (in \cite{BBCOS}, thanks to the explicit form of the effective potential, we could use the stronger bound $|\nabla^m \widehat{V}_\text{eff} (k)| \lesssim |k|^{-1}$, for $|k| \geq 1$). 

Following the strategy of \cite{BBCOS}, in the first part of the proof we bound the Fourier coefficients $\widehat{s}_k$ and their (discrete) derivatives $\delta_j^m \widehat{s}_k$. We find (proceeding as in the proof of \cite[Lemma 2.2]{BBCOS}, recalling that $\widehat{V}_\text{eff} (0) = 8 \pi \frak{a} > 0$ and therefore that $\widehat{V}_\text{eff} (k) \geq c > 0$, for $k$ small enough) 
\begin{align}
\label{eq:sfourier_I}
    |\delta_j^m \widehat{s}_k| \leq C_m \, |k|^{-m} \min \Big\{  \frac{\rho^{1/4}}{|k|^{1/2}} , \frac{\rho}{k^2} \Big\} 
\end{align}
for all $m\in \mathbb N$ and $k\in \Lambda^*_+$ with $|k|\leq 1$. For $|k| \geq 1$, we get $|\delta_j^m \widehat{s}_k| \leq C_m \rho / k^2$ for all $m \in \bN$ (this should be compared with \cite[(A.8)]{BBCOS}, where we could prove the stronger estimate $|\delta_j^m \widehat{s}_k| \leq C_m \rho / |k|^3$, using the decay of $\widehat{V}_\text{eff}$). For $|k| \geq 1$, it is also important to derive the improved estimate 
\begin{align}
\label{eq:sfourier_II}
   \left| \delta_j^m\left[ \widehat{s}_k+\frac{\rho_0 \widehat{V}_\mathrm{eff}(k)}{2|k|^2}-\frac{\rho_0^2 \widehat{V}_\mathrm{eff}(k)^2}{2|k|^4}+\frac{11 \rho_0^3 \widehat{V}_\mathrm{eff}(k)^3}{16|k|^6} \right]\right| & \leq C_m \frac{\rho^4}{|k|^8} 
\end{align}
for all $m\in \bN$. Compared with \cite[(A.9)]{BBCOS}, here it is important to expand $\widehat{s}_k$ to higher order, to compensate the lack of decay of $\widehat{V}_\text{eff}$. 

Next, we use the bounds for the Fourier coefficients $\widehat{s}_k$ to estimate $s$ in position space. In particular, we show that 
\begin{equation}\label{eq:s-point} | s (x)| \leq C \frac{\rho}{|x|} \min \Big\{ 1 , \frac{1}{(\rho^{1/2} |x|)^{3/2}} \Big\} \,, \qquad |\nabla s (x)| \leq C \frac{\rho |\hspace{-.05cm} \log \rho |}{x^2} \end{equation} 
for all $|x| \geq  15 r_0$, and that, for $r\geq 15r_0$
\begin{equation}\label{eq:int-s}\begin{split} 
    \int_{|x|\geq r}  |\nabla s(x)|^2dx  &\lesssim \min\left\{r^{-4}\rho^{\frac{1}{2}},r^{-1}\rho^2\right\},\\
       \int_{|x|\geq r}   \left|\Delta s (x) \right|^2dx& \lesssim   r^{-3}\rho^2,\\
     \int_{|x|\geq r}  \left|\nabla \Delta s (x) \right|^2dx& \lesssim   r^{-5}\rho^2.
     \end{split} 
\end{equation}

To show (\ref{eq:s-point}) and (\ref{eq:int-s}), we decompose $s$ into a low- and a high-momentum component. To this end, we choose $\chi \in C_c^\infty (\bR^3)$ with $\chi (z) = 1$ for $|z| \leq 2$ and $\chi (z) = 0$, for $|z| \geq 4$. Moreover, for $t > 0$, we set $\chi_t (z) = \chi (z/t)$ and $s_t = \widecheck{\chi}_t * s$. With this notation, we write $s = s_1 + (s-s_1)$. The contribution of the low-momentum part $s_1$ to (\ref{eq:s-point}), (\ref{eq:int-s}) can be controlled exactly as in \cite[Lemma 2.2]{BBCOS} (at low-momenta, the analysis relies on (\ref{eq:sfourier_I})); we skip the details. To estimate the high momentum component $s-s_1$, on the other hand, we have to modify the approach of \cite{BBCOS}. We introduce $h = - \rho_0 \chi_{r_0} / 8\pi |.| \in L^2 (\bR^3)$ and we define $D' \in L^2 (\bR^3)$ through its Fourier transform 
\begin{align*}
    \widehat{D}'(k):=\widehat{V}_\mathrm{eff}(k)\widehat{h}(k).
\end{align*}
Recalling that $V_\mathrm{eff}$ is supported in $B_{r_0}(0)$, 
we find, for (almost every) $|x|>r_0$, 
\begin{align}
\label{Eq:explicit_formula_D(x)}
    D'(x) = -\frac{\rho_0}{8\pi}\int  \frac{\chi_{r_0}(x-y)}{|x-y|}V_\mathrm{eff}(y)dy.
\end{align}
In Fourier space, at high momenta, $D'$ is an approximation for the factor $- \rho_0 \widehat{V}_\text{eff} (k) / 2k^2$ appearing in the expansion of $\widehat{s}_k$ (see (\ref{eq:sfourier_II})). More precisely, we find 
\begin{align*}
    \widehat{D}'_k +\frac{\rho_0 \widehat{V}_\mathrm{eff}(k)}{2|k|^2}=\lim_{R\rightarrow \infty} \frac{\rho_0 \widehat{V}_\mathrm{eff}(k)}{8\pi}\int \frac{\chi_R(x)-\chi_{r_0}(x)}{|x|} \, e^{i k \cdot x}dx.
\end{align*}
From the smoothness of $(\chi_R(x) - \chi_{r_0}(x))/|x|$, in particular from the estimate \[ \int \left|\nabla^\alpha \left(\frac{\chi_R(x)-\chi_{r_0}(x)}{|x|}\right)\right| dx\leq C \] holding for all $\alpha  \geq 3$, uniformly in $R>0$, we obtain, for any $\alpha,m\in \mathbb N$, 
\begin{equation}\label{eq:D'+}
\left|\delta^m_j \left[\widehat{D}'_k +\frac{\rho_0 \widehat{V}_\mathrm{eff}(k)}{2|k|^2}\right]\right|
\lesssim \rho |k|^{-\alpha}
\end{equation}
for every $|k| \geq 1$. From \eqref{eq:sfourier_II}, we conclude that 
\[ \Big| \delta_j^m \Big[ \widehat{s}_k - \widehat{D}'_k - 2 \widehat{D}^{'2}_k - \frac{11}{2} \widehat{D}^{'3}_k \Big] \Big| \lesssim \frac{\rho}{|k|^8} \,. \]
To approximate the high momentum part $s-s_1$, we set $D := D' - \widecheck{\chi}_1 * D'$ (in momentum space $\widehat{D}_k = (1- \chi_1 (k)) \widehat{D}'_k$). We define the error term $E$, requiring that 
\begin{align}
\label{Eq:decomposition_s_high_momenta}
   s-s_1 = D + 2 D*D + \frac{11}{2}D*D*D + E \,.
\end{align}
Observing that $\widehat{E}_k = \widehat{s}_k - \widehat{D}'_k - 2 \widehat{D}^{'2}_k - 11 \widehat{D}^{'3}_k /2$ for $|k| \geq 4$ (because $\widehat{D}_k = \widehat{D}'_k$, for $|k| \geq 4$) and using (\ref{eq:D'+}) and $|\delta_j^m \widehat{s}_k| \lesssim \rho / k^2$ for $1 \leq |k| \leq 4$, we conclude that  
\begin{equation}\label{eq:deltaE} |\delta_j^m \widehat{E}_k| \lesssim \frac{\rho}{|k|^8} \end{equation}
for all $|k| \geq 1$. Furthermore, $\widehat{E}_k = 0$, for $|k| \leq 1$. With the estimate 
 \begin{equation*}
        |x_j^m f(x)|\le \Big(\frac{\pi}{2}\Big)^m\Big| \Big(\frac{L}{2\pi}\Big)^m\big(e^{-i\frac{2\pi}{L}x_j}-1\big)^m f(x)\Big|=\Big(\frac{\pi}{2}\Big)^m  \big| (\widecheck{\delta^m_j \widehat{f}} ) (x)\big|\leq \Big(\frac{\pi}{2}\Big)^m \left\|\delta^m_j \widehat{f} \right\|_1 
    \end{equation*}
we can translate (\ref{eq:deltaE}) to position space. We find the pointwise bounds
\[ |E(x)|, |\nabla E(x)| , |\Delta E (x)| , |\nabla \Delta E (x)| \lesssim \frac{\rho}{|x|^m}\]
for any $m \geq 1$ (the bound for $\nabla \Delta E$ requires $k \to \delta_j^m |k|^3 \widehat{E} (k)$ to be summable; that's why we had to expand $s$ to third order). To estimate $s-s_1$, we still have to control the contributions $D, D*D, D*D*D$ appearing on the r.h.s. of  \eqref{Eq:decomposition_s_high_momenta}. Since $V_\mathrm{eff}$ is supported in $B_{r_0}(0)$ and since the function $\chi_{r_0}$ appearing in \eqref{Eq:explicit_formula_D(x)} is supported in $B_{4r_0}(0)$, we conclude that $D'$ is supported in $B_{5r_0}(0)$. Hence, for $|x| \geq 5r_0$ and $q\in \mathbb N$, we find that
    \begin{align*}
       | \nabla^{q} D(x) | = | -\left(\nabla^q \widecheck{\chi}_{1}\right)*D' | = \Big| \frac{\rho_0}{8\pi}\int  \frac{\chi_{r_0}(x-y)}{|x-y|}(\nabla^q \widecheck{\chi}_{1} * V_\mathrm{eff})(y) dy \Big| \lesssim \frac{\rho}{|x|^m}
    \end{align*}
for any $m \in \bN$. Similarly, we can estimate $\nabla^q (D*D)$ and $\nabla^q (D*D*D)$, for $|x|\geq 10r_0$ and, respectively, for $|x|\geq 15r_0$. By \eqref{Eq:decomposition_s_high_momenta}, in combination with the pointwise bounds for $E$ and its derivatives, we therefore obtain, for $|x|\geq 15r_0$, 
    \begin{align*}
        \left|(s-s_1)(x)\right|, \left|\nabla (s-s_1)(x)\right|, \left|\Delta(s-s_1)(x)\right|, \left|\nabla \Delta (s-s_1)(x)\right|\lesssim \frac{\rho}{|x|^{m}} 
    \end{align*}
for any $m > 0$. Combining these bounds with the estimates for the low-momentum component $s_1$ (established as in \cite{BBCOS}) we obtain (\ref{eq:s-point}), (\ref{eq:int-s}). The rest of the proof is identical to \cite[Lemma 2.2]{BBCOS} (the bounds (\ref{eq:int-s}) play an important role in the proof of (\ref{eq:decay})).

\end{document}